\documentclass[10pt,journal,compsoc]{IEEEtran}
\IEEEoverridecommandlockouts
\usepackage{cite}
\usepackage{amsmath,amssymb,amsfonts}
\usepackage{algorithmicx}
\usepackage{graphicx}
\usepackage{textcomp}
\usepackage{xcolor}
\def\BibTeX{{\rm B\kern-.05em{\sc i\kern-.025em b}\kern-.08em
    T\kern-.1667em\lower.7ex\hbox{E}\kern-.125emX}}
\definecolor{BrickRed}{RGB}{178,34,34} 

\usepackage{booktabs}
\usepackage{algorithm}
\usepackage[noend]{algpseudocode}
\usepackage{listings}
\usepackage[utf8]{inputenc}
\usepackage{tikz} 
\usepackage[T1]{fontenc}
\usepackage{subfigure} 
\usepackage{graphicx}
\usepackage[mathlines,switch]{lineno}

\usetikzlibrary{%
  arrows,%
  calc,
  shapes,
  arrows,
  shapes.misc,
  shapes.arrows,
  chains,%
  matrix,%
  positioning,
  scopes,%
  decorations.pathmorphing,
  shadows%
}

\usepackage{relsize}  

\usepackage{amsthm}
\newtheorem{theorem}{Theorem}
\newtheorem{lemma}{Lemma}
\newtheorem{definition}{Definition}
\usepackage{soul}
\soulregister\cite7 
\soulregister\citep7 
\soulregister\citet7 
\soulregister\ref7 
\soulregister\pageref7 

\makeatletter
\def\ps@IEEEtitlepagestyle{%
  \def\@oddfoot{\mycopyrightnotice}%
  \def\@evenfoot{}%
}
\def\mycopyrightnotice{%
  {\footnotesize 
  \begin{minipage}{\textwidth}
  \centering
  \textbf{This work has been submitted to the IEEE for possible publication. Copyright may be transferred without notice, after which this version may no longer be accessible.}
  \end{minipage}
    \hfill}
  \gdef\mycopyrightnotice{}
}

\begin{document}
%
\title{BLOWN: A Blockchain Protocol for Single-Hop Wireless Networks under Adversarial SINR}
%
%
%
%

\author{Minghui~Xu,~\IEEEmembership{Member,~IEEE,}
    Feng~Zhao,~\IEEEmembership{Member,~IEEE,}
    Yifei~Zou, 
    Chunchi~Liu, 
    Xiuzhen~Cheng,~\IEEEmembership{Fellow,~IEEE,}
    Falko~Dressler,~\IEEEmembership{Fellow,~IEEE}
\thanks{M. Xu, Y. Zou, and X. Cheng are with the School of Computer Science and Technology, Shandong University, Qingdao, 266510, P. R. China. E-mail: \{mhxu,yfzou,xzcheng\}@sdu.edu.cn}
\thanks{F. Zhao (Correspondign Author) is with the Guangxi Colleges and Universities Key Laboratory of Complex System Optimization and Big Data Processing, Yulin Normal University, Yulin, P.R. China. E-mail: zhaofeng@guet.edu.cn.}
\thanks{C. Liu is with Ernst \& Young, Shanghai, 200120, P. R. China. E-mail: peter.cc.liu@cn.ey.com.}
\thanks{F. Dressler is with the the School of Electrical Engineering and Computer Science, TU Berlin, 10587 Berlin, Germany. E-mail: dressler@ccs-labs.org.}

\thanks{Manuscript created June 20, 2020}}

\IEEEtitleabstractindextext{%
\begin{abstract}
Known as a distributed ledger technology (DLT), blockchain has attracted much attention due to its properties such as decentralization, security, immutability and transparency, and its potential of servicing as an infrastructure for various applications. Blockchain can empower wireless networks with identity management, data integrity, access control, and high-level security. However, previous studies on blockchain-enabled wireless networks mostly focus on proposing architectures or building systems with popular blockchain protocols. Nevertheless, such existing protocols have obvious shortcomings when adopted in wireless networks where nodes may have limited physical resources, may fall short of well-established reliable channels, or may suffer from variable bandwidths impacted by environments or jamming attacks. In this paper, we propose a novel consensus protocol named Proof-of-Channel (PoC) leveraging the natural properties of wireless communications, and develop a permissioned BLOWN protocol (BLOckchain protocol for Wireless Networks) for single-hop wireless networks under an adversarial SINR model. We formalize BLOWN with the universal composition framework and prove its security properties, namely persistence and liveness, as well as its strengths in countering against adversarial jamming, double-spending, and Sybil attacks, which are also demonstrated by extensive simulation studies.
\end{abstract}

}

\maketitle

\IEEEdisplaynontitleabstractindextext

%
\IEEEpeerreviewmaketitle

\ifCLASSOPTIONcompsoc
\IEEEraisesectionheading{\section{Introduction}\label{sec:introduction}}
\else
\section{Introduction}
\label{sec:introduction}
\fi
\IEEEPARstart{D}istributed Ledger Technology (DLT) refers to share, replicate, and synchronize a digital ledger across a distributed network without centralized data storage. As a widely used DLT, blockchain technologies intend to organize a digital ledger as a chain of blocks to enable remarkable properties such as decentralization, immutability, and traceability. Since Bitcoin has emerged as the first open cryptocurrency, blockchain has been envisioned as a promising technology that can be used in various practical applications such as finance \cite{tapscott2017blockchain}, Internet of Things (IoT) \cite{ali2018applications}, supply chain \cite{korpela2017digital}, and security services \cite{salman2018security}. 
In recent years, the popularity of 5G and IoT has arisen more problems of managing devices, sharing information, and carrying on computing tasks among wireless nodes \cite{soret2015interference}. Such problems become even intractable in a wireless network with small-world and super-dense features \cite{cheng2017ssdnet}. To overcome these challenges, researchers have been making continuous effort to build secure and trusted computing environments such as mobile edge computing enabled blockchain \cite{xiong2018mobile} and the blockchain empowered 5G \cite{dai2019blockchain} in wireless networks taking advantage of blockchain technologies.
As shown in Fig.~\ref{Fig:uav}, one of the most typical application scenarios of wireless blockchain is a single-hop unmanned vehicle network \cite{ghribi2020secure, calvo2018secure, aloqaily2021design}, in which blockchain-based unmanned vehicles can execute precise cooperative operations (by consensus mechanism) based on trusted historical information (using blockchain as a decentralized ledger).  This system can be fault-tolerant, robust, and secure against malicious attacks.


\begin{figure}[htbp!]
\centering
\includegraphics[width=0.45\textwidth]{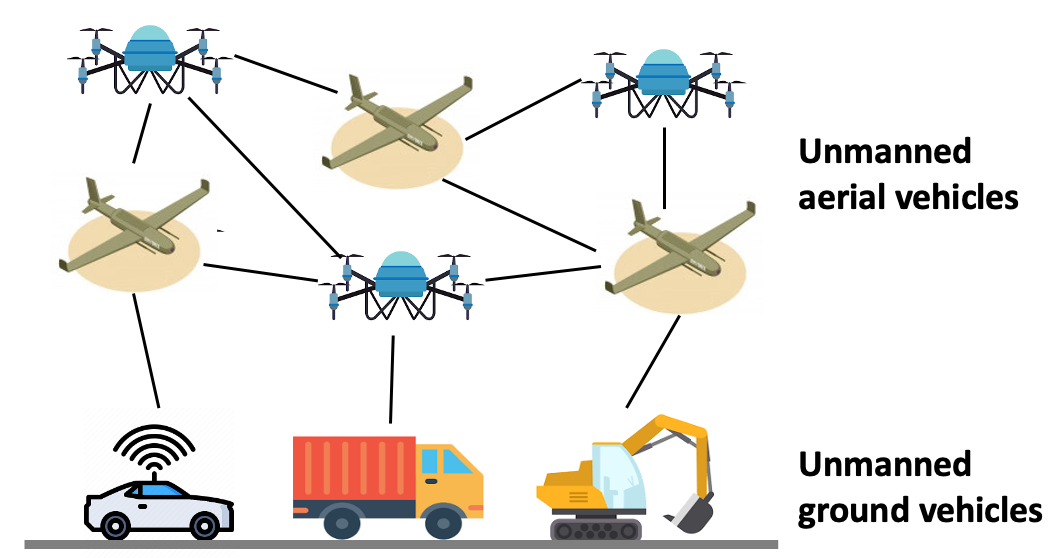}
\caption{Unmanned Vehicles Network.}
\label{Fig:uav}
\end{figure}
Previous studies on blockchain-enabled wireless networks mostly focus on proposing architectures or building systems on top of popular blockchain protocols that are previously deployed on the Internet. Such blockchain protocols make use of consensus algorithms that are based on either proof of resources or  message passing. 
Proof of resources based consensus requires users to compete for proposing blocks by demonstrating their utilization of physical resources such as energy and storage (e.g. Proof-of-Work  \cite{Bitcoin}) or virtual resources such as reputation and weight (e.g., Proof-of-Stake \cite{Ouroboros}). 
Message passing based protocols such as PBFT \cite{pbft}, on the other hand, require the participants to reach consensus through message exchanges. 
Even though these consensus algorithms perform well for existing blockchain protocols, they are not suitable for wireless networks since they are mainly developed for systems with Internet serving as the underlying network infrastructure. The reasons can be concluded as follows: 1) Wireless networks fall short of well-established reliable channels built with physical wires such as fiber as the Internet does -- the open free air communications are severely impacted by environments (e.g., interference or contention) resulting in the variable channel bandwidths and latency. This poses threatens to blockchain consensus process; 2) Even though some of the exiting blockchain protocols do not require strong network synchrony, which means that they operate properly when the transmission delay is bounded, they still need the support of basic media access control protocols (e.g., Carrier sense multiple access with collision avoidance (CSMA/CA)) in wireless networks. CSMA/CA is inefficient to address heavy contention in dense networks, and can cause additional traffic. 3) Wireless networks are particularly vulnerable to jamming attacks. However, existing blockchain protocols fall short of defending jammers efficiently. 
These barriers make it very possible for communications to fail, causing the traditional consensus algorithms inapplicable. Such problems are not sufficiently addressed by existing blockchain protocols, which motivates our study on blockchain over wireless. 

In this paper, we propose BLOWN, a BLOckchain protocol for Wireless Networks, to overcome the above challenges. BLOWN is a two-phase protocol that adopts a new concept, namely Proof-of-Channel (PoC), to seamlessly  integrate the procedures of  consensus and channel competition.  In PoC, nodes compete for available channels to win the rights of proposing blocks. Such a design makes probing the wireless channel conditions part of the consensus procedure,  successfully reducing the communication cost while increasing consensus efficiency and effectiveness.  On the other hand, we consider that an adversary can make adversarial jamming on the nodes  but controls no more than 50\% wealth of the network in BLOWN, where 
wealth is defined to be the total number of coins held by all users. BLOWN is a provably secure system that satisfies two formal security properties: persistence and liveness. Persistence means that if an honest node proclaims a transaction as stable, other honest nodes, if queried, either report the same result or report error messages. Liveness,  on the other hand, states that the transactions originated from the honest nodes can eventually be added to the blockchain. To prove BLOWN's properties, we formally model it with a universally composable (UC) framework and analyze it accordingly. 
Note that it is worthy of emphasizing that PoC can be adapted to multi-hop wireless networks if combined with existing techniques such as distributed spanner construction \cite{yu2017distributed, xu2021wchain}, or supported by an adequate routing layer \cite{awerbuch2002an}.

Our main contributions are summarized as follows.
\begin{enumerate}
\item To the best of our knowledge, BLOWN is the first provably secure protocol that is specifically designed for single-hop wireless networks under a realistic adversarial SINR model. 
\item A novel, general Proof-of-Channel consensus protocol is proposed in this paper, which leverages the natural properties of wireless networks such as broadcast communications and channel competitions. 
\item We develop a UC-style protocol for BLOWN and formally prove BLOWN's persistence and liveness properties by showing that it satisfies concrete chain growth, common prefix, chain quality properties.
\item Finally, extensive simulation studies are conducted to validate our theoretical analysis.
\end{enumerate}

The rest of the paper is organized as follows. Section~\ref{sec:related:Work} introduces the most related works on state-of-the-art blockchain protocols. Section~\ref{sec:model} presents our models and assumptions. In Section~\ref{sec:blown}, the two-phase BLOWN protocol is explained in detail. Security properties of BLOWN are analyzed in Section~\ref{sec:security:analysis}. We report the results of our simulation studies in Section~\ref{sec:simulations} and conclude this paper in Section~\ref{sec:conclusion}.

\section{Related Work}
\label{sec:related:Work}

\textbf{Blockchain consensus protocols.} We classify blockchain consensus protocols into two categories: proof of resources (virtual or physical) and message passing, and overview state-of-the-arts in this section. For a more comprehensive survey we refer the readers to \cite{DBLP:journals/comsur/XiaoZLH20}.  

Proof of physical resources requires that users compete for proposing blocks by demonstrating their utilization of physical resources. Proof-of-Work (PoW) is of the most use in blockchain. The most popular example of PoW-based blockchain is the Bitcoin proposed in 2008, which selects leaders by mining power \cite{Bitcoin}. Ethereum provides the Turning-complete Ethereum Virtual Machine (EVM) and adopts a modified PoW (with Ethash) \cite{wood2014ethereum}. Free-pool mining~\cite{shi2020fee} was proposed for PoW to incentivize miners to behave cooperatively. Alternatives to PoW include Proof-of-Space \cite{dziembowski2015proofs}, Proof-of-Burn (PoB) \cite{proofofburn}, Proof-of-Elapsed Time (PoET) \cite{Hyperledger},  etc., in which Proof-of-Space, also known as Proof-of-Capacity or Proof-of-Storage, refers to consensus nodes competing by occupied memories or disk spaces, PoB means that a node can destroy coins to virtually earn mining rights, and PoET, proposed by Intel, leverages trusted hardware (e.g., SGX) to determine how long a node has to wait before it is allowed to generate a block.

In contrast to proof of physical resources, proof of virtual resources aims to show the utilization of virtual resources such as reputation, stake, or elaborately defined weight. For example, Proof of Stake (PoS) was developed to address the power consumption issue of PoW and it resorts to stakes as voting rights rather than computational powers. Algorand uses a cryptographic Sortition algorithm to randomly select a verifiable committee according to stakes \cite{gilad2017algorand}. IOHK created the Ouroboros family in recent years, which adopts PoS and G.O.D coin tossing to randomly choose a leader according to  stakes  \cite{Ouroboros}. Snow White utilizes an epoch-based committee which embodies successful block miners in a specific time period so that all nodes have an identical view of the committee \cite{bentov2016snow}. In Proof-of-Reputation (PoR), each node is assigned a reputation \cite{miller2014permacoin}, and a node can write blocks only when its reputation meets certain requirements; thus PoR always comes with incentive mechanisms or economic penalty strategies.  

In message passing based blockchain protocols, nodes can perform local computations and broadcast messages to each other to reach consensus. This method provides blockchain the robustness to Byzantine failures while ensuring liveness and safety. In Ripple, a transaction that receives more than 80\% votes from UNL-recorded servers can step into the next round, and transactions having survived through the whole RPCA process can be added to the blockchain \cite{Ripple}. ELASTICO partitions nodes by their unique identities and a consensus is reached in each shard based on byzantine agreement protocols \cite{elastico}. Stellar creates overlapped shards, also known as quorum slices, leveraging Federated Byzantine Agreement (FBA) to reach consensus \cite{Stellar}. Omniledger uses lottery-like RandHound and VRF-based leader election algorithms to assign validators to each shard \cite{Omniledger}. Other message-passing based protocols utilized in blockchain include PBFT \cite{pbft}, HoneyBadgerBFT \cite{miller2016honey}, Tendermint \cite{kwon2014tendermint}, Hotstuff \cite{yin2019hotstuff}, and CloudChain \cite{xu2021cloudchain}. 

\textbf{Blockchain for Internet of Things.} IoT encompasses devices that are generally connected to a wireless network.  Blockchain has been applied for various IoT applications such as access management, security enhancement, and privacy protection. 
Novo developed a new architecture, which contains six components, for access management in IoT based on blockchain \cite{novo2018blockchain}. 
Dorri \textit{et al.} optimized blockchain for IoT by introducing a distributed trust model, in which new blocks proposed by the users with high trust can be free from complete transaction validation to decrease the transaction processing overhead. 
Feng \textit{et al.} \cite{feng2020joint} proposed a radio and computational resource allocation joint optimization framework for  blockchain-enabled  mobile edge computing. 
In vehicular ad hoc networks, Malik \textit{et al.} \cite{DBLP:journals/winet/MalikNHL20} utilized blockchain to achieve secure key management. 
%
Guo \textit{et al.} \cite{guo2020network} presented a novel endogenous trusted framework for IoT, which integrates blockchain, software-defined networking, and network function virtualization. Guo \textit{et al.}  \cite{guo2019blockchain} constructed a blockchain-based authentication system to realize trusted data sharing among heterogeneous IoT platforms. In \cite{LiuTMCTokoin}, Liu \emph{et al.} developed  a tokoin (token+coin) based novel framework to provide fine-grained and accountable access procedure control leveraging blockchain for various IoT applications. Its unique significance lies in that the fine-grained access policy can be defined, modified, and revoked only by the resource owner and  the whole access procedure, not just the access action alone,  can be accountably and securely controlled.    In \cite{LiuTCTrust}, Liu \emph{et al.} proposed an important idea of extending trust from on-chain to off-chain making use of trusted hardware and blockchain to ensure that the data in digital world is consistent with the truth in physical world and that any change in physical world should be reflected instantly by the data change in digital world.

None of the works mentioned above considers the properties of wireless communications when designing their blockchain protocols. To our best knowledge, wChain \cite{xu2021wchain} presented in 2021, a blockchain protocol designed for multi-hop wireless networks, is the most relevant one.  wChain constructs a spanner-based communication backbone on a multi-hop wireless network, making use of a fault-tolerant consensus without involving the underlying physical wireless layer. Hence wChain can be complementary to BLOWN when BLOWN needs to be migrated to multi-hop networks and realize fault tolerance. 

\textbf{Consensus protocols for wireless networks.} Since consensus is the core of blockchain and our study is closely related to wireless networks, we briefly survey the studies on consensus protocols for wireless networks. 
The abstract MAC layer \cite{kuhn2009abstract} is one of the earliest models that can achieve elegant abstraction and precisely capture the fundamental guarantees of the low-level network layer performance in wireless networks. 
Newport  provided the first tight bound for distributed consensus in radio networks \cite{newport2014consensus}. With the abstract MAC layer, Newport and Robinson gave a fault-tolerant consensus algorithm that terminates within $O(N^3\log N)$, where $N$ is the unknown network size \cite{newport2018fault}.
A pioneering work on the implementation of the abstract MAC layer provides a groundbreaking scheme to adaptively tune the contention and interference in wireless channels \cite{yu2018exact}. 
Moniz \textit{et al.} \cite{moniz2012Byzantine} proposed a BFT consensus protocol allowing $k>\lfloor \frac{N}{2} \rfloor$ faulty nodes with time complexity of $O(N^2)$. They assumed an abstract physical layer in wireless ad hoc networks and directly used high-level broadcast primitives. 
Chockler \textit{et al.} \cite{chockler2005consensus} explored fault-tolerant consensus with crashed nodes. Their study reveals the relationship of collision detection and fault-tolerant consensus under a graph-based model. 
Assuming realistic message delays and a graph model, Scutari and Sergio designed a distributed consensus algorithm for wireless sensor networks \cite{scutari2008distributed}, making use of a network model that considers the MAC layer with a multipath and frequency-selective channel.
Aysal \textit{et al.} \cite{aysal2009reaching} studied the average consensus problem with probabilistic broadcasts. They explored the effect of wireless medium on the consensus process and extended the non-sum preserving algorithm to accelerate convergence. 

\textbf{Summary.}
A common drawback of proof of physical resources lies in their prohibitively large demands of physical resources such as high computational power, storage, energy, or specific hardware, of which devices in wireless networks are notoriously limited; on the other hand, proof of virtual resources might encounter centralization problems caused by the over-powerful validators or authorities. 
Even though honest nodes without high power would not harm a blockchain, it is possible for a malicious node (e.g., an honest node that has been hacked) to launch attacks freely if we do not take any measure to restrict the over-powerful nodes.
Additionally,  message-passing protocols always incur a significant amount of message exchanges leading to non-trivial communication overhead.  Existing message-passing protocols need to exchange at least $O(N)$ messages for consensus. When being applied in wireless settings, these protocols need the support of basic wireless networking functions such as CSMA/CA for contention control. If considering the message overhead of the underlying MAC protocols, their message overhead is even higher, especially in dense wireless networks. Most notably, almost all existing works mentioned above were developed for the Internet resting on the closed medium (e.g., fiber) with sufficient bandwidth where jamming is not an issue. However, existing blockchain protocols are vulnerable to jamming attacks in wireless networks. 

Motivated by these observations, in this paper we propose BLOWN, a wireless blockchain protocol that relies on a  newly-developed PoC to seamlessly integrate  wireless communications with  blockchain consensus  while guaranteeing persistence and liveness, the two critical security properties of blockchain, to counter jamming and Sybil attacks.

\section{Models and Assumptions}
\label{sec:model}

\textbf{Network Model.} In this paper, we consider a network with a set $V$ of $N$ nodes arbitrarily deployed in a communication space. Such a network could contain a group of manipulated Unmanned Aerial Vehicles (UAVs) or intelligent robots in realistic scenarios. A node is equipped with a half-duplex transceiver that can transmit or receive a message, or sense the channel, but cannot transmit and receive or transmit and sense simultaneously.  Let $d(u,v)$ be the Euclidean distance between nodes $u$ and $v$, $D_R(v)$ denote the disk centered at $v$ with a radius $R$, and $N_R(v)$ denote the set of nodes within $D_R(v)$ including $v$.  The notations of $D_R(v)$ and $N_R(v)$ are further utilized in the definition of single-hop network and the protocol analysis.
%

We assume that each node knows the identities, locations and public keys of all other nodes. We further assume that each node can generate key pairs and has access to a secure EUF-CMA digital signature scheme (details of cryptographical tools employed in this paper are presented in protocol analysis and simulation sections). Each node maintains a hash-chain of blocks, and each block contains multiple transactions. We denote frequently-used notations of transaction, block, blockchain, chain of block headers by $tx$, $B$, $BC$, and $BH$, respectively, and use super/subscript to attach more specific information. A transaction is modeled as a coin exchange process. We adopt the notion of the unspent transaction outputs (UTXOs) accounting method. A UTXO-based account stores coins in a set of UTXOs, and a UTXO defines an output of a blockchain transaction that has not been spent. This UTXO model provides a high level of security since it is convenient to authenticate all transaction sequences using UTXOs, limiting the risk of double-spending attacks. 

\begin{figure}[htbp!]
\centering
\includegraphics[width=0.45\textwidth]{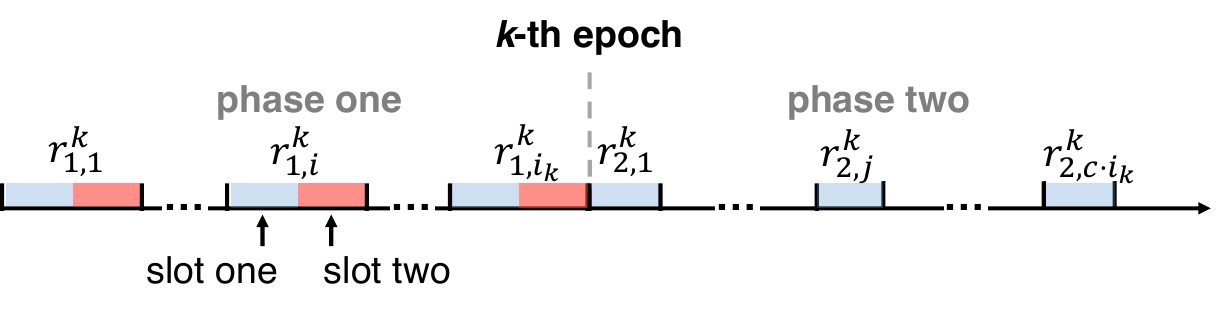}
\caption{Epoch-based execution.}
\label{Fig:epoch}
\end{figure}

\textbf{Interference Model.} We adopt the Signal-to-Interference-plus-Noise-Ratio (SINR) wireless network model, which captures the network interference in a more realistic and precise manner than a graph-based one \cite{yu2020competitive}. A standard SINR model can be formulated as follows, which states that a message sent by $u$ is correctly received by $v$ if and only if
\begin{equation}
\label{eq:sinr}
\begin{aligned}
SINR=\frac{\mathcal{S}}{\mathcal{I}+\mathcal{N}}\geq \beta
\end{aligned}
\end{equation}
where $\mathcal{S}=P\cdot d(u,v)^{-\alpha}$ is the received signal power of node $v$ from $u$, $P$ is the uniform transmit power, $\mathcal{I}=\sum_{w\in W\setminus\{u\}}P\cdot d(w,v)^{-\alpha}$ is the interference at $v$ and $W$ is the set of nodes that transmit in the current round, $\mathcal{N}$ is the ambient noise, $\alpha\in (2,6]$ is the path-loss exponent, and threshold $\beta>1$ is determined by hardware. To capture the fine-grained noise, we define $\mathcal{N}=\mathcal{ADV}(v)$, where $\mathcal{ADV}(v)$ is the composite noise generated by the environment and adversaries. Assume each node uses a common noise threshold $\theta$.   Since we consider a single-hop network where all nodes are within the transmission range of each other,  the distance between any two nodes is bounded by $R_0=(P/\beta\theta)^{1/\alpha}$. We further assume that each node can perform physical carrier sensing. If at least one neighboring node $u$ of $v$ broadcasts a message, $v$ would either receive a message or sense a busy channel. At each slot, a node $v$ may either (a) sense an idle channel (the measured signal power is lower than $\theta$); (b) receive a message (the measured signal power exceeds $\theta$ and $SINR\geq\beta$); or (c) sense a busy channel (the measured signal power exceeds $\theta$ but $SINR<\beta$).
Let $RSS=\mathcal{S}+\mathcal{I}+\mathcal{N}$ be the total received signal power at a node. Then when the node receives a message, the interference plus noise can be calculated by $\mathcal{I}+\mathcal{N}=RSS-\mathcal{S}$ with a known $\mathcal{S}$ \cite{ikuno2012accurate,zhao2011efficient,olszewski2007sinr,jeske2007method,abu1999determining}.
Besides,  nodes are not required to be fully synchronized as assuming that when a node transmits, all other nodes can receive (correctly decode) the message. We only require partial synchronization that a node may not be able to receive a message due to channel contention, but it can sense a busy channel when another node transmits the message.

\textbf{Epoch-based Execution.} 
As shown in Fig.~\ref{Fig:epoch}, the BLOWN protocol is executed in disjoint and consecutive time intervals called \textit{epochs}, and at each epoch no more than one block can be generated. Each epoch ${e_k}$ consists of two phases with each containing multiple rounds. In $e_k$, we denote $r^{k}_{1,i}$ as the $i$-th round in phase one and $r^{k}_{2,j}$ as the $j$-th round in phase two, with $r^{k}_{1,i}$ consisting of two slots and $r^{k}_{2,j}$ having only one slot. Besides, $i_k$ is the length of phase one and $c\cdot i_k$ is the length of phase two, where $c$ is a variable constant determined later. If a node just join or reconnect the blockchain network, it can synchronize blocks and history from its peers and then normally execute the BLOWN protocol.

\textbf{Adversary.} Honest nodes strictly follow the BLOWN protocol.  Besides, we assume that there exists a group of adversaries who can freely join or leave the network, create identities, or make noises to interfere with any honest node at any time. For simplicity, the group of adversaries can be regarded as a powerful adversary $\mathcal{A}$ who controls less than 50\% wealth of the entire network. $\mathcal{A}$ can launch jamming attacks by continuously sending messages without following the protocol or even colluding with other jammers. To leave a chance for an honest node to communicate, $\mathcal{A}$ is ($(1-\epsilon),T$)-bounded at any time interval of length $T$ rounds, where $T \in \mathbb{N}$ (the set of natural numbers) and $0<\epsilon\leq 1$, indicating that the super adversary can jam nonuniformly at most $(1-\epsilon)T$ rounds within $T$. Each node $v$ maintains a variable $T_v$, which is the estimate of $T$ by $v$.  

In this paper, we say that event $E$ occurs with high probability (w.h.p.) if for any $c\geq1$, $E$ occurs with probability at least $1-1/N^c$, and with moderate probability (w.m.p.) if for any $c\geq1$, $E$ occurs with probability at least $1-1/\log^cN$. A summary of all important notations (including the ones from the BLOWN protocol and the protocol analysis) and their semantic meanings is provided in Table~\ref{table:summary:notation}.

\begin{table}
\caption{Summary of Notations}
\label{table:summary:notation}
\centering
\begin{tabular}{c|c}
  \hline
  \textbf{Symbol} & \textbf{Description} \\
  \hline
  $B^{k}_v$ & block generated by node $v$ in the $k$-th epoch\\
  $BC_v^k$ & blockchain locally stored at node $v$\\
  $BH_v^k$ & block headers read from $BC_v^k$ \\
  $P_1(P_2)$ & phase one (phase two) \\
  $V$ & set of all nodes \\
  $N$ & network size \\
  $c$ & constant used to determine the length of phase two\\
  $c_v$ & counter variable used to record round information\\
  $i_k$ & length of $P_1$ of the $k$-th epoch\\
  $tx$ & a transaction \\
  $txp_v$ & temporary transaction stack of node $v$\\
  $p_v$ & $v$'s probability of sending a message\\
  $p_V$ & aggregated probability of all nodes\\
  $\hat{p}$ & an upper bound of $p_v$\\
  $r^{k}_{1,i}(r^{k}_{2,j})$ & $i$-th ($j$-th) round of $P_1(P_2)$ in the $k$-th epoch \\
  $l_v$ & $l_v\in \{0,1,\cdots, w_v\}$ as the leader counter of $v$\\
  $l_v^0$ & initial value of $l_v$ generated by \emph{Sortition}()\\
  $T$ & time window of the adversary\\
  $T_v$ & estimate of $T$ by node $v$\\
  $w_v$ & deposit of node $v$\\
  $\pi_v$ & proof created by \emph{Sortition}()\\
  $\epsilon$ & proportion of non-jammed rounds\\
  $\tau$ & the parameter to determine the hardness of Sortition\\
  \hline
\end{tabular}
\end{table}

\section{The BLOWN protocol}
\label{sec:blown}
In this section, we present the two-phase BLOWN protocol. We first summarize the BLOWN protocol by providing an overview of BLOWN and its construction primitives, and then detail the protocol itself.

\subsection{Overview and Utilities of BLOWN}
In this subsection, we present an overview on BLOWN, and describe the construction primitives/utilities to more precisely and concisely illustrate the BLOWN protocol.

\subsubsection{An Overview on BLOWN}


\begin{figure}[htbp!]
\centering
\includegraphics[width=0.5\textwidth]{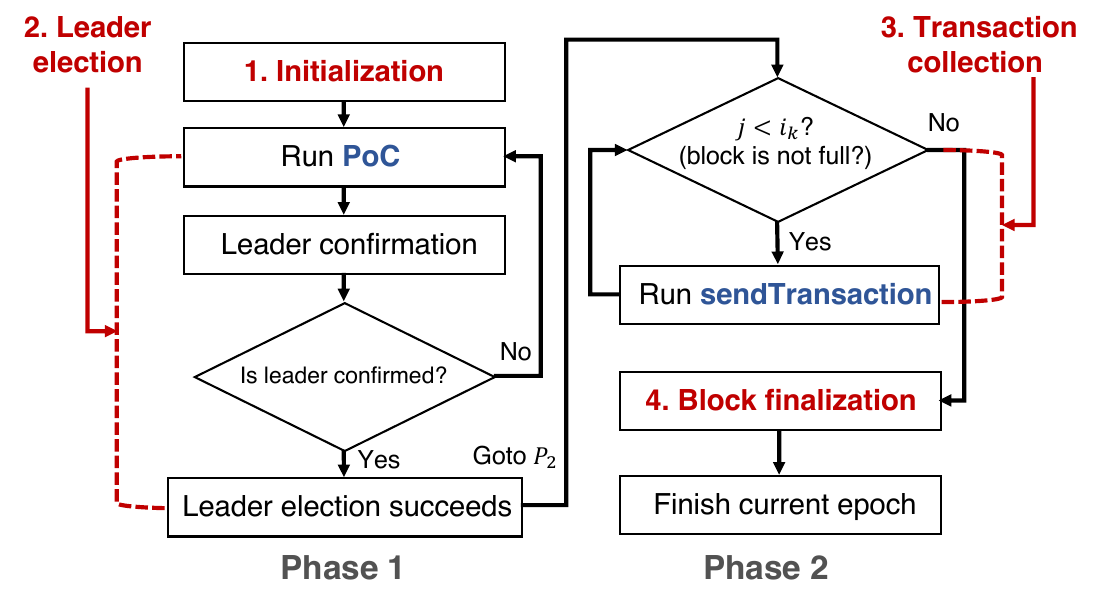}
\caption{BLOWN protocol overview.}
\label{Fig:overview}
\end{figure}

\begin{algorithm}[!ht]
\caption{Utilities for node $v$}\label{alg:utility}
  \begin{algorithmic}[1]
  \Function{\emph{Sortition}} {$sk_v$, seed, role, $\tau$, $w_v$, $W$}
      \State $h_v$, $\pi_v$ = VRF($sk_v$, seed$||$role)
      \State $p=\tau$/$W$, $l_v=$\emph{LeaderCounter}(role, $w_v$, $h_v$, $p$)
      \State \textbf{return} $h_v, \pi_v, l_v$
  \EndFunction

  \Function{\emph{VerifySortition}} {$pk_v$, $h_v$, $\pi_v$, seed, role, $\tau$, $w_v$, $W$, $l_v$} 
      \If {VerifyVRF($pk_v$, $h_v$, $\pi_v$, seed$||$role) $= 0$} 
          \State \textbf{return} 0
      \EndIf
      \State $p=\tau$/$W$, $\hat{l_v}=$\emph{LeaderCounter}(role, $w_v$, $h_v$, $p$)
      \If {$\hat{l_v}\neq l_v$}
          \State \textbf{return} 0
      \EndIf
      \State \textbf{return} 1
  \EndFunction

  \Function{\emph{LeaderCounter}} {role, $w_v$, $h_v$, $p$}
      \State $l_v=0$
      \If {role is \texttt{FOLLOWER}}
          \State \textbf{return} $l_v$
      \EndIf
      \If {$\frac{h_v}{2^l} \in \left[0, \sum_{k=0}^{l_v} B(k;w_v,p)\right]$}
          \State \textbf{return} $l_v$
      \EndIf
      \State  $l_v=1$
      \While {$\frac{h_v}{2^l} \notin \left( \sum_{k=0}^{l_v-1} B(k;w_v,p), \sum_{k=0}^{l_v} B(k;w_v,p)\right]$}
          \State $l_v=l_v+1$
      \EndWhile
      \State \textbf{return} $l_v$
  \EndFunction

  \Function{$\emph{MSG}$} {$r^{k}_{1,i}$, $l_v$}
      \State $m.r^{k}_{1,i}\leftarrow r^{k}_{1,i}$ 
      \State $m.l_v\leftarrow l_v$ 
      \State \textbf{return} $m$
  \EndFunction

  \Function{$\emph{MSGT}$} {$tx$, $r^{k}_{2,j}$, $l_v$}
      \State $m_T.tx\leftarrow tx$
      \State $m_T.r^{k}_{2,j}\leftarrow r^{k}_{2,j}$ 
      \State $m_T.l_v\leftarrow l_v$ 
      \State \textbf{return} $m_T$
  \EndFunction

  \Function{$\emph{MSGB}$} {$BC_v^{k}$, $B_v^{k}$, $r^{k}_{2,j}$, $l_v$, role, $w_v$, $h_v$, $\pi_v$, $l_v'$}
      \State $m_B.BH_v^{k}\leftarrow$ read $BH_v^{k}$ from $BC_v^{k}$
      \State $m_B.B_v^{k}\leftarrow B_v^{k}$
      \State $m_B.r^{k}_{2,j}\leftarrow r^{k}_{2,j}$ 
      \State $m_B.l_v\leftarrow l_v$
      \State $m_B.sort\leftarrow \{\text{role}, w_v, h_v, \pi_v, l_v'\}$
      \State \textbf{return} $m_B$
  \EndFunction

  \Function{\emph{Pack}} {$txp_v$}\Comment{Pack $txp$ to form a block}
      \State \textbf{return} $B_v^{k}$
  \EndFunction

  \Function{\emph{Append}} {$BC_v^{k-1}$, $B_u^{k}$}\Comment{append the new block} 
      \State \textbf{return} $BC_v^{k}\leftarrow BC_v^{k-1}+B_u^{k}$
  \EndFunction

  \end{algorithmic}
\end{algorithm}

The BLOWN protocol proceeds in epochs, with each constructing no more than one block. Specifically, our protocol has two phases within an epoch, denoted by $P_1$ and $P_2$ as shown in Fig.~\ref{Fig:overview}. $P_1$ is responsible for initialization and leader election while $P_2$ is for transaction collection and block finalization. In our design, nodes contend by broadcasting messages on a wireless channel. In response, we establish a robust jamming-resistant channel by introducing an adaptive transmission mechanism, confronting background noise and jamming simultaneously. Such a channel is realized by dynamically adjusting the transmission probability $p_v$ of each node $v$ according to its sensed contention in the network.  Concretely, we first adopt the Sortition algorithm to assign $v$ a weight $l_v$ based on its account balance. Sortition ensures that splitting coins to generate massive identities cannot break our protocol. After initialization, the protocol starts the process of leader election. We utilize the nature of contention in a wireless network to design our proof-of-channel consensus mechanism (PoC). To achieve usability and efficiency, PoC allows nodes to compete on the channel right-of-use to obtain opportunities of proposing blocks rather than rely on extra physical resources or introduce communication overhead. More specifically, upon receiving a message, $l_v$ is decremented. The sole survivor with non-zero $l_v$ at the end of $P_1$ is appointed as the leader. This essentially integrates leader election and channel contention into a single process, namely the phase one of BLOWN. In $P_2$, the leader is responsible for collecting and verifying transactions, assembling them into a new block, and then broadcasting the block to the whole network. If the new block is valid, it is admitted by all honest nodes.

\subsubsection{Utilities}
\label{sub:sec:utilities}
Algorithm~\ref{alg:utility} lists the following frequently used functions for any node $v$ in BLOWN: \emph{Sortition}(), \emph{VerifySortition}(), \emph{LeaderCounter}(), \emph{MSG}(), \emph{MSGT}(), \emph{MSGB}(), \emph{Pack}(), and \emph{Append}(). It also presents the following data structures employed by the above functions: transaction $tx$, transaction stack $txp_v$, block $B^{k}_v$, block header $BH_v^{k}$, blockchain $BC_v^{k}$, basic message $m$, transaction message $m_T$, and block message $m_B$.

Concretely, we modify the crytographic Sortition algorithm proposed by Algorand\cite{gilad2017algorand} to make it suitable for our BLOWN protocol. The Sortition algorithm is based on a verifiable random function (VRF), which takes as inputs a private key $sk_v$, a random seed and a role, and outputs a hash $h_v$ as well as its corresponding proof $\pi_v$. There are two types of roles: a \texttt{FOLLOWER} who can only be a follower during an epoch and a \texttt{\texttt{LEADER}} who is a potential leader. Besides, $W$ is the accumulated number of coins of all users in the network, $w_v$ is the deposit of node $v$, $l_v \in \{0,1,\cdots,w_v\}$ is the leader counter of the node $v$, and $p=\frac{\tau}{W}$ is the probability based on which each coin is used to increment the counter value where $\tau$ determines the hardness. The probability of $l_v=k$ follows the binomial distribution $B(k;w_v,p)=\binom{w_v}{k}p^k(1-p)^{w_v-k}$ with $\sum_0^{w_v} B(k;w_v,p)=1$. To determine $l_v$, the \emph{LeaderCounter}(role, $w_v$, $h_v$, $p$) divides $[0,1]$ into consecutive intervals as $I(l_v) = \left[0, \sum_{k=0}^{l_v} B(k;w_v,p)\right]$ for ${l_v=0}$ and $I(l_v)=\left( \sum_{k=0}^{l_v-1} B(k;w_v,p), \sum_{k=0}^{l_v} B(k;w_v,p)\right]$ for $I(l_v) \in \{1,\cdots,w_v\}$. If $v$'s role is \texttt{FOLLOWER}, $I(l_v)=0$; otherwise, if the normalized hash $\frac{h_v}{2^l}$ ($l$ is the hash length) falls in the interval $I(l_v)$, $l_v$ is returned as the value of the leader counter. The function \emph{VerifySortition}() intends to check if $h_v, \pi_v, l_v$ are valid by calling \emph{VerifyVRF}() and recomputing \emph{LeaderCounter}().

Three functions, namely \emph{MSG}(), \emph{MSGT}(), and \emph{MSGB}(), generate messages that can respectively be used for leader election, transaction collection, and block finalization. Specifically, \emph{MSG}() creates a basic message $m$ for leader election in $P_1$, \emph{MSGT}() produces a message $m_T$ embodying a transaction, which is sent during the transaction collection process in $P_2$, and \emph{MSGB}() outputs a message $m_B$ which contains a $B_v^{k}$ generated by the leader $v$, a $BH_v^{k}$ read from $BC_v^{k}$, the current value $l_v$ of the leader counter, and a string $\{\mbox{role}, w_v, h_v, \pi_v, l_v'\}$ used to verify Sortition where $l_v'$ is the original value of the leader counter. To reduce communication cost, we send $BH_v^{k}$ embodied in $m_B$ for a simplified verification. Finally, \emph{Pack}($txp_v$) is adopted to validate and pack transactions to form a new block, and \emph{Append}($BC_v^{k-1}$, $B_u^{k}$) appends the new block $B_u^{k}$ to the local blockchain $BC_v^{k-1}$.

\subsection{The BLOWN Protocol Specifications}
\label{protocol}
In a nutshell, BLOWN is a two-phase protocol. As shown respectively in Algorithm~\ref{alg:round1} and Algorihtm~\ref{alg:round2}, phase $P_1$ is employed for initialization and leader election while phase $P_2$ is for transaction collection and block finalization.

\subsubsection{Phase $P_1$}

\begin{algorithm}[!htbp]
  \caption{PoC subroutine}\label{alg:phase1:sub}
  \begin{algorithmic}[1]
        \If {$v$ decides to send a message based on $p_v$}
            \State $m\leftarrow$ $\emph{MSG}$($r^{k}_{1,i}$, $l_v$), $v$ broadcasts $(m,\sigma)$
        \Else
            \If {channel is idle} 
                \State $p_v=\min\{(1+\gamma)p_v, \hat{p}\}$
                \State $T_v = \max\{1, T_v-1\}$
            \Else
                \If {$v$ receives a message $(m,\sigma)$}
                    \State $p_v=(1+\gamma)^{-1}p_v$
                    \State $l_v=l_v-1$ 
                \EndIf
            \EndIf
        \EndIf
        \State $c_v = c_v + 1$
            \If {$c_v \geq T_v$}
                \State $c_v = 1$
                \If {there is no idle rounds in the past $T_v$ rounds}
                    \State $p_v = (1+\gamma)^{-1}p_v$,
                    \State $T_v = T_v + 2$
                \EndIf
            \EndIf
  \end{algorithmic}
\end{algorithm}

\begin{algorithm}[!htbp]
  \caption{BLOWN $P_1$ protocol}\label{alg:round1}
  \begin{algorithmic}[1]
    \State $\triangleright$ \textcolor{BrickRed}{Initialization}
    \State $h_v,\pi_v,l_v^0=$\emph{Sortition}($sk_v$, seed$||$role, $\tau$, $w_v$, $W$)
    \State $p_v = \hat{p}$, $c_v=0$, $T_v=1$, $i=1, l_v=l_v^0$
    \State $\triangleright$ \textcolor{BrickRed}{Leader election}
    \While {TRUE}
    \If {$l_v>0$} \Comment{\textcolor{blue}{As a potential leader}}
        \State $\triangleright$ \textcolor{blue}{slot one of $r^{k}_{1,i}$}
        \State run \textbf{PoC}

    \State $\triangleright$ \textcolor{blue}{slot two of $r^{k}_{1,i}$}
    \If {$v$ broadcasts a message in slot one}
        \State $v$ listens on the channel
        \If {$v$ senses an idle channel}
            \State Goto $P_2$ with $i_k=i$ \Comment{run $P_2$ as a leader}
        \EndIf
    \Else 
        \State $m\leftarrow$ $\emph{MSG}$($r^{k}_{1,i}$, $l_v$), and $v$ broadcasts $(m,\sigma)$
    \EndIf

   \Else \Comment{\textcolor{blue}{As a follower}}
        \State $\triangleright$ \textcolor{blue}{slot one of $r^{k}_{1,i}$}
        \State $v$ listens on the channel to receive a message

        \State $\triangleright$ \textcolor{blue}{slot two of $r^{k}_{1,i}$}
        \If {in slot one $v$ receives $(m,\sigma)$ from $u$ and has $\mathcal{I+N}<\theta$}
            \If {$v$ senses an idle channel}
                \State $v$ recognize $u$ as the leader
                \State Goto $P_2$ with $i_k=i$ \Comment{run $P_2$ as a follower}
            \EndIf
        \Else
            \State $m\leftarrow$ $\emph{MSG}$($r^{k}_{1,i}$, $l_v$), and $v$ broadcasts $(m,\sigma)$
        \EndIf
    \EndIf

    \State $i=i+1$
  \EndWhile
  \end{algorithmic}
\end{algorithm}

Let's examine the details of the BLOWN $P_1$ protocol. Lines 2-3 of Algorithm~\ref{alg:round1} constitute the initialization process. First, \emph{Sortition}() takes as inputs $sk_v$, seed$||$role, $\tau$, $w_v$ and $W$ (see Line 2), and outputs $h_v$, $\pi_v$ and $l_v^0$, where $h_v$ and $\pi_v$ are respectively a hash and its corresponding proof, and $l_v^0 \in \{0,1,\cdots,w_v\}$ stands for the initial leader counter.  All the inputs of the Sortition algorithm are illustrated in Section.~\ref{sub:sec:utilities}.  Note that $l_v>0$ indicates that $v$ remains to be a potential leader while $l_v=0$ indicates that $v$ is a follower. Let $\hat{p}$ be the maximum transmission probability, which can be initialized to any small number in $(0,1)$. Since the absence of followers might lead to a bad case in which all nodes are potential leaders and simultaneously broadcast messages in slot one\footnote{Such a bad case only occurs with a small probability, which is less than $\hat{p}^n$}, we prevent this from occurring by ensuring that there always exist at least one follower after initialization. A simple approach to achieving this goal is to artificially and randomly add followers (with a \texttt{FOLLOWER} rule) to the network.
Second, we set $p_v = \hat{p}$, $c_v=0$, $T_v=1$, $i=1, l_v=l_v^0$ (Line 3 ), where $p_v$ is the probability on which node $v$ decides to send a message, and is upper-bounded by $\hat{p}$, $c_v=1$ is a counter variable used to record round information, $T_v$ is the estimate of the time window of the adversary by node $v$, $i$ is the round counter used in $P_1$, and $l_v$ is the leader counter variable initialized to $l_v^0$. 
After initialization, $P_1$ proceeds round by round with each containing two slots, and a node $v$'s activity at each slot depends on its role. 

Before proceeding any further, we need to explain the PoC subroutine described in Algorithm~\ref{alg:phase1:sub} to adjust $l_v$, $p_v$ and $T_v$, the leader counter, transmission probability and adversary's time window estimate, according to the sensed channel condition at the first slot of each round in $P_1$. Specifically, $v$ with $l_v>0$ (a potential leader) performs the following actions: it either broadcasts a message $(m,\sigma)$ with probability $p_v$ (Lines 1-2), where $\sigma$ is the signature of $m$, or senses the channel with probability $1-p_v$ (Lines 3-10). One can see that $v$ adapts its $p_v$ in a multiplicative increase or decrease manner by a factor of $(1+\gamma)$, where $\gamma=O(1/(\log T+\log \log N))$ is a small number that is loosely determined by $N$ and $T$ (see the proof of Theorem~\ref{thm:jamming:resistant}).
More specifically, $p_v$ is multiplicatively increased (Line 5) when the channel is sensed idle or decreased (Line 9) when a message is received\footnote{Receving a message indicates the message has a valid signature, and we do not explicitly present the signature verification process for conciseness.}. Such a mechanism ensures that honest nodes can cooperatively adjust their transmission probabilities to help reduce contention on the channel. Meanwhile, we decrease $T_v$ by 1 if the channel is idle (Line 6) as the estimate of adversary's time window seems to be too large when the channel is idle, and decrease $l_v$ by 1 if a neighbor of $v$ successfully broadcasts a message (Line 10) as the neighbor seems to have a better chance of being the leader. On the other hand, if the number of rounds in $P_1$ is no less than $T_v$ (Line 12), the estimate of the adversary's time window, we further check whether or not there is an idle round in the past $T_v$ rounds (Line 14), and if not, $p_v$ is decreased (Line 15) and $T_v$ is increased (Line 16) to further adjust $p_v$ and $T_v$. 
One can see that a successful broadcast causes the decrements of the $l_v$ values of the receivers. When $l_v=0$, $v$ becomes a follower who can only sense the channel in slot one of the next round. We establish a robust jamming-resistant channel by introducing an adaptive transmission mechanism, confronting channel contention and jamming attacks simultaneously.  Such a channel is realized by dynamically adjusting the transmission probability $p_v$ of each node $v$ according to its sensed contention in the network.   This mechanism can better address jamming attacks 
compared to the Carrier-sense multiple access with collision avoidance (CSMA/CA) technique.

\begin{figure}[htbp!]
\centering
\includegraphics[width=0.5\textwidth]{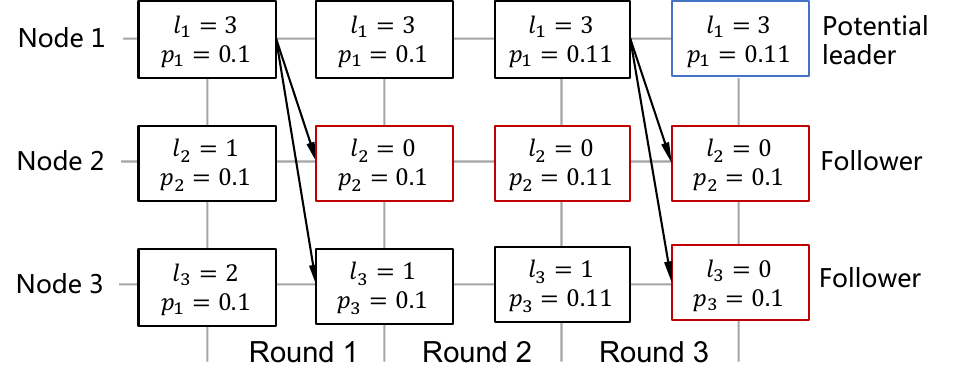}
\caption{A toy example of running PoC with three nodes.}
\label{Fig:poc}
\end{figure}

To better illustrate the PoC subroutine, we provide a toy example with three nodes. Assume that $\gamma=0.1$ and initially $p_1=p_2=p_3=\hat{p}=0.1$ and $l_1=3, l_2=1,l_3=2$. In the first round, node 1 successfully transmits a message to node 2 and 3, thus $l_2$ and $l_3$ each decreases by 1 and $p_2=p_3=\min\{(1+\gamma)p_v, \hat{p}\}=0.1$. Since $l_2=0$,  node 2 becomes a follower. In the second round, all nodes choose not to transmit so they all sense an idle channel and increment $p_1, p_2, p_3$ to $0.1\times(1+\gamma)=0.11$. In the third round, node 1 successfully transmits to node 2 and 3. Then node 3 becomes a follower and only node 1 survives as the unique potential leader, at which time $p_2=p_3=0.11\times(1+\gamma)^{-1}=0.1$.

Now we are back to continue explaining Phase $P_1$ of the BLOWN protocol, which contains multiple rounds. At slot one of each round,  if $v$ is a potential leader, which means $l_v>0$, $v$ runs the PoC subroutine described in Algorithm~\ref{alg:phase1:sub} (Line 8); otherwise, $v$ listens on the channel for message reception (Line 18).  At slot two of each round, $v$ behaves according to its actions in slot one. If $v$ as a potential leader broadcasts a message in slot one and senses an idle channel in slot two, it can set itself as a leader and goto $P_2$ (Lines 10-13); otherwise it broadcasts a message in slot two. A follower $v$ recognizes $u$ as the leader only when $v$ believes $u$ is the only transmitter in slot one and senses an idle channel in slot two (Lines 20-23); otherwise $v$ transmits in slot two (Line 25). 
In Theorem \ref{thm:leader:success}, we prove that slot two is capable of letting the leader and the followers mutually recognize each other. 

At the end of $P_1$, there should be only one survivor with $l_v>0$, who then becomes the leader. Note that $i_k$ denotes the length of $P_1$, which is used to determine the length of $P_2$. We will prove in Theorem~\ref{thm:leader:success} of Section~\ref{sec:security:analysis} that Algorithm~\ref{alg:round1} ensures a successful leader election.

\begin{algorithm}[!htbp]
  \caption{sendTransaction subroutine}\label{alg:phase2:sub}
  \begin{algorithmic}[1]
        \If {$v$ decides to send a message based on $p_v$}
                \State $m_T$ $\leftarrow$$\emph{MSGT}$($tx$, $r^{k}_{2,j}$, $l_v$), and $v$ broadcasts $(m_T,\sigma_T)$
            \Else
                \If {channel is idle} 
                    \State $p_v=\min\{(1+\gamma)p_v, \hat{p}\}$
                    \State $T_v =  \max\{1, T_v-1\}$
                \Else
                    \If {receives a message $(m_T,\sigma_T)$}
                        \State $p_v=(1+\gamma)^{-1}p_v$
                    \EndIf
                \EndIf
        \EndIf

            \State $c_v = c_v + 1$
                \If {$c_v \geq T_v$}
                    \State $c_v = 1$
                    \If {there is no idle round in the past $T_v$ rounds}
                        \State $p_v = (1+\gamma)^{-1}p_v$,
                        \State $T_v = T_v + 2$
                    \EndIf
                \EndIf
  \end{algorithmic}
\end{algorithm}

\subsubsection{Phase $P_2$}

\begin{algorithm}[!htbp]
  \caption{BLOWN $P_2$ protocol}\label{alg:round2}
  \begin{algorithmic}[1]
    \State $\triangleright$ \textcolor{BrickRed}{Transaction collection}
    \While {$j<c\cdot i_k$}
        \If {$l_v>0$} \Comment{\textcolor{blue}{As a leader}}
            \State $v$ listens on the channel to receive a $(m_T,\sigma_T)$
            \If {receives $m_T$.tx$\neq\perp$}
                \State $txp_v$[$j$] = $m_T.tx$
            \EndIf
        \Else \Comment{\textcolor{blue}{As a follower}}
            \State run \textbf{sendTransaction}
        \EndIf
        \State $j=j+1$
\EndWhile

    \State $\triangleright$ \textcolor{BrickRed}{Block finalization}
    \If {$j=c\cdot i_k$}
    \If {$l_v>0$} \Comment{\textcolor{blue}{As a leader}}
        \State $B_v^{k}\leftarrow\emph{Pack}$($txp_v$)
        \State $BC_v^{k}\leftarrow\emph{Append}$($\textsc{BC}_v^{k-1}$, $B_v^{k}$)
        \State $m_B\leftarrow\emph{MSGB}$($BC_v^{k}$, $B_v^{k}$, $r^{k}_{2,j}$, $l_v$, role, $w_v$, $h_v$, $\pi_v$, $l_v^0$), and broadcasts $(m_B,\sigma_B)$
    \Else \Comment{\textcolor{blue}{As a follower}}
        \If {receives $(m_B,\sigma_B)$ \&\& \emph{VerifySortition} ($pk_v$, seed, $\tau$, $W$, $m_B.sort$) $=1$}
            \State \emph{Append}($\textsc{BC}_v^{k-1}$, $m_B.B_u^{k}$)
        \EndIf
    \EndIf
    \EndIf
  \end{algorithmic}
\end{algorithm}

Phase $P_2$ of BLOWN performs transaction collection and block finalization, as shown in Algorithm~\ref{alg:round2}. It proceeds by a fixed amount of $c\cdot i_k$ rounds where each round contains only one slot, and $c$ is a constant to directly determine the length of $P_2$ and indirectly the maximum block size, which can be adjusted according to specific implementations. We refer to $j$ as the round counter in $P_2$. If $j<c\cdot i_k$, a leader selected in $P_1$ should listen to the channel to receive signed transaction messages, which are recorded in the stack $txp_v$, while other nodes continuously broadcast signed transaction messages (Lines 2-8). After $c\cdot i_k$ rounds, the leader serializes all transactions to form a new block denoted by $B_v^{k}\leftarrow$\emph{Pack}($txp_v$), and broadcasts the $(m_B, \sigma_B)$ (Lines 12-15). Once receiving a $(m_B,\sigma_B)$ from $u$, a node $v$ should append the new block to its local blockchain only if $\sigma_B$ is valid and \emph{VerifySortition} ($pk_v$, seed, $\tau$, $W$, $m_B.sort$) $=1$ (Lines 17-18). Note that the sendTransaction subroutine presented in Algorithm~\ref{alg:phase2:sub} is employed by $P_2$ to broadcast transactions and the parameters $p_v$, $c_v$, $T_v$ are utilized to ensure jamming-resistant communications as they function in the PoC subroutine shown in Algorithm~\ref{alg:phase1:sub}.

\section{Protocol Analysis}
\label{sec:security:analysis}
\begin{figure*}[!htbp]
  \centering
  \setlength{\fboxsep}{0.3cm}
  \fbox{
    \begin{minipage}{1.95\columnwidth}
      \begin{center}
        \textbf{Protocol $\mathcal{\pi}_{\rm B}[\mathcal{F}_{\rm SIG}, \mathcal{F}_{\rm SORT}]$}
      \end{center}
      \vspace{0.1cm}

      $\mathcal{\pi}_{\rm B}$ is a protocol run by all nodes interacting with the ideal functionalities $\mathcal{F}_{\rm SIG}$ and $\mathcal{F}_{\rm SORT}$.

      \textbf{Initialization:} Send ($sk_v$, $seed$, $role$, $\tau$, $w_v$, $W$) to $\mathcal{F}_{\rm SORT}$, which returns ($h_v,\pi_v,l_v^0$). Next, initialize the remaining local parameters as $p_v = \hat{p}$, $c_v=0$, $T_v=1$, $i=1, l_v=l_v^0$.

      \textbf{Leader election:} For each round $r^{k}_{1,i}$ of $P_1$ during the $k$-th epoch, perform the following (1) or (2) according to the value of $l_v$:

      \begin{enumerate}
      \renewcommand{\labelenumi}{(\theenumi)}
          \item If $l_v>0$, run PoC in slot one. If broadcasting a message in slot one, listen on the channel in slot two and if the channel is idle, goto $P_2$ with $i_k=i$ at the end of slot two; otherwise, send $m\leftarrow \emph{MSG}$($r^{k}_{1,i}$, $l_v$) to $\mathcal{F}_{\rm SIG}$, which returns a signed message $(m, \sigma)$, i.e., $(m, \sigma)$ is obtained by querying $\mathcal{F}_{\rm SIG}$, then broadcast $(m, \sigma)$ in slot two. 
          \item If $l_v=0$, listen on the channel in slot one. If receiving a valid $(m, \sigma)$ from $u$ with $\mathcal{I+N}<\theta$ in slot one, and sensing an idle channel in slot two, recognize $u$ as the leader and goto $P_2$ with $i_k=i$ at the end of slot two; otherwise, generate $m\leftarrow$ $\emph{MSG}$($r^{k}_{1,i}$, $l_v$), send $m$ to $\mathcal{F}_{\rm SIG}$, which returns $(m, \sigma)$, then broadcast $(m, \sigma)$ in slot two. Note that a valid $m$ holds when $\mathcal{F}_{\rm SIG}$ returns  1 upon being queried with ($m, \sigma$).  
      \end{enumerate}

      \textbf{Transaction collection:} At each round $r^{k}_{2,j}$, if $l_v>0$, listen on the channel for a possible signed transaction message $m_T$, add the transaction to the local stack as $txp_v$[$j$] = $m_T.tx$ if receiving $(m_T, \sigma_T)$ and $\mathcal{F}_{\rm SIG}$ returns 1 when being queried with $(m_T, \sigma_T)$. 
      If $l_v=0$, send $m_T$ $\leftarrow$$\emph{MSGT}$($tx$, $r^{k}_{2,j}$, $l_v$) to $\mathcal{F}_{\rm SIG}$, which returns a signed message $(m_T, \sigma_T)$, then broadcast $(m_T, \sigma_T)$. 

      \textbf{Block finalization:} During the round $r^{k}_{2,c\cdot i_k}$, if $l_v>0$, execute $B_v^{k}\leftarrow\emph{Pack}$($txp_v$) and $BC_v^{k}\leftarrow\emph{Append}$($\textsc{BC}_v^{k-1}$, $B_v^{k}$); then generate $m_B\leftarrow\emph{MSGB}$($BC_v^{k}$, $B_v^{k}$, $r^{k}_{2,j}$, $l_v$, role, $w_v$, $h_v$, $\pi_v$, $l_v^0$) and send it to $\mathcal{F}_{\rm SIG}$, which returns $(m_B, \sigma_B)$. 
      If $l_v=0$, listen on the channel for a possible block message $(m_B, \sigma_B)$; if receiving a valid $(m_B, \sigma_B)$, which means $\mathcal{F}_{\rm SIG}$ returns 1 upon being queried with ($m_B, \sigma_B$), and $\mathcal{F}_{\rm SORT}$ returning 1 upon being queried with ($pk_v$, seed, $\tau$, $W$, $m_B.sort$), execute \emph{Append}($\textsc{BC}_v^{k-1}$, $m_B.B_u^{k}$).

    \end{minipage}
    }
  \caption{The protocol (hybrid experiment) $\mathcal{\pi}_{\rm B}[\mathcal{F}_{\rm SIG}, \mathcal{F}_{\rm SORT}]$.}
  \label{fig:uc:blowm}
\end{figure*}

Proving security properties of a complex protocol such as BLOWN is very challenging. Thus we leverage the universally composable (UC) framework proposed by Canetti \textit{et al.} \cite{959888}. The UC framework captures the security of a protocol via emulating an idealized protocol $\mathcal{F}$ (often referred to as an ideal functionality), which satisfies strong security properties. Then a real protocol $\pi$ specifying concrete implementations is said to be secure if it is indistinguishable from $\mathcal{F}$. The main feature of the UC framework is the universal composability that allows one to perform analysis on a complex protocol, whose security properties can be derived from the security of its components. 

\begin{figure*}[!htbp]
  \centering
  \setlength{\fboxsep}{0.3cm}
  \fbox{
    \begin{minipage}{1.95\columnwidth}
      \begin{center}
        \textbf{Protocol $\pi_{\rm SORT}[\mathcal{F}_{\rm VRF}]$} 
      \end{center}
      
      \vspace{0.3cm}
    
      \textbf{LeaderCounter:} When activated with input (role, $w_v$, $h_v$, $p$), first initialize $l_v=0$. If role is \texttt{FOLLOWER}, output $l_v=0$ and exit. If role is \texttt{LEADER}, compute $\frac{h_v}{2^l}$; if $\frac{h_v}{2^l}$ falls in $\left[0, \sum_{k=0}^{l_v} B(k;w_v,p)\right]$, output $l_v=0$ and exit; otherwise, increase $l_v$ until it satisfies that $\frac{h_v}{2^l} \in \left( \sum_{k=0}^{l_v-1} B(k;w_v,p), \sum_{k=0}^{l_v} B(k;w_v,p)\right]$, then send $l_v$ to $v$ and output $l_v$.

      \textbf{Sortition:} When activated with input ($sk_v$, seed, role, $\tau$, $w_v$, $W$), first feed ($sk_v$, seed$||$role) to $\mathcal{F}_{\rm VRF}$, which returns ($h_v$, $\pi_v$); then compute $p=\tau$/$W$ and input (role, $w_v$, $h_v$, $p$) to LeaderCounter, which returns $l_v$; finally, output ($h_v, \pi_v, l_v$).
      
      \textbf{VerifySortition:} When activated with input ($pk_v$, $h_v$, $\pi_v$, seed, role, $\tau$, $w_v$, $W$, $l_v$), first feed ($pk_v$, $h_v$, $\pi_v$, seed$||$role) to $\mathcal{F}_{\rm VRF}$, which returns ($pk_v$, $h_v$, $\pi_v$, seed$||$role, $f$). If $f=0$, output \texttt{FALSE}, which means that verification fails; if $f=1$, compute $p=\tau$/$W$ and feed ($h_v$, $\pi_v$) to LeaderCounter to obtain $\hat{l_v}$. Following that, if $\hat{l_v}\neq l_v$, output \texttt{FALSE}; otherwise output \texttt{TRUE}, which means that verification succeeds. 
      
    \end{minipage}
    }
  \caption{The protocol (hybrid experiment) $\pi_{\rm SORT}[\mathcal{F}_{\rm VRF}]$.}
  \label{fig:uc:sort}
\end{figure*}

\begin{figure*}[!htbp]
  \centering
  \setlength{\fboxsep}{0.3cm}
  \fbox{
    \begin{minipage}{1.95\columnwidth}
      \begin{center}
        \textbf{Functionality $\mathcal{F}_{\rm SORT}$} 
      \end{center}
      
      \vspace{0.3cm}
      \textbf{LeaderCounter:} Upon receiving (role, $w_v$, $h_v$, $p$) from some node $v$, verify if role is \texttt{FOLLOWER}. If so, send $l_v=0$ to $v$; otherwise, compute $\frac{h_v}{2^l}$. Next if $\frac{h_v}{2^l}$ falls in $\left[0, \sum_{k=0}^{l_v} B(k;w_v,p)\right]$, sends $l_v=0$ to $v$; otherwise increase $l_v$ until it satisfies that $\frac{h_v}{2^l} \in \left( \sum_{k=0}^{l_v-1} B(k;w_v,p), \sum_{k=0}^{l_v} B(k;w_v,p)\right]$, then send $l_v$ to $v$.

      \textbf{Sortition:} Upon receiving ($sk_v$, seed, role, $\tau$, $w_v$, $W$) from some node $v$, send ($sk_v$, seed$||$role) to the adversary, who returns ($h_v$, $\pi_v$). 

      \begin{enumerate}
      \renewcommand{\labelenumi}{(\theenumi)}
      \item If there is no entry ($sk_v$, seed$||$role, $h_v$, $\pi_v$) recorded, record ($sk_v$, seed$||$role, $h_v$, $\pi_v$); if there is an existing entry ($sk_v$, seed$||$role, $h_v'$, $\pi_v'$) that satisfies $h_v'=h_v$ and $\pi_v'=\pi_v$, do nothing. Next compute $p=\tau$/$W$ and send (role, $w_v$, $h_v$, $p$) to LeaderCounter, which returns $l_v$. Finally, send ($h_v, \pi_v, l_v$) to $v$. 
      \item If there is an entry ($sk_v$, seed$||$role, $h_v'$, $\pi_v'$) recorded but $h_v'\neq h_v$ or $\pi_v'\neq \pi_v$, send an error message to $v$. 
      \end{enumerate}
      
      \textbf{VerifySortition:} Upon receiving ($pk_v$, $h_v$, $\pi_v$, seed, role, $\tau$, $w_v$, $W$, $l_v$), send ($pk_v$, $h_v$, $\pi_v$, seed$||$role) to the adversary, who returns ($pk_v$, $h_v$, $\pi_v$, seed$||$role, $f$). 

      \begin{enumerate}
      \renewcommand{\labelenumi}{(\theenumi)}
      \item If $f=0$ or there is no entry ($sk_v$, seed$||$role, $h_v$, $\pi_v$) recorded, send $0$ to $v$, which means that verification fails. 
      \item If $f=1$ and there is an existing entry ($sk_v$, seed$||$role, $h_v$, $\pi_v$), compute $p=\tau$/$W$ and send ($h_v$, $\pi_v$) to LeaderCounter, which returns $\hat{l_v}$. If $\hat{l_v}\neq l_v$, sends 0 to $v$, i.e., verification fails; otherwise send 1 to $v$ meaning that verification succeeds. 
      \end{enumerate}

    \end{minipage}
    }
  \caption{The ideal functionality $\mathcal{F}_{\rm SORT}$
  \label{fig:uc:ideal:sort}}
\end{figure*}

\subsection{UC Composition of BLOWN}
\label{sec:sub:uc:composition}
We formulate two UC-style protocols (or hybrid experiments), which are presented in Fig.~\ref{fig:uc:blowm} and Fig.~\ref{fig:uc:sort}. 
The $\mathcal{\pi}_{\rm B}[\mathcal{F}_{\rm SIG}, \mathcal{F}_{\rm SORT}]$ conducts a hybrid experiment for BLOWN using an ideal hybrid functionality $[\mathcal{F}_{\rm SIG}, \mathcal{F}_{\rm SORT}]$ where $\mathcal{F}_{\rm SIG}$ is an ideal digital signature scheme and $\mathcal{F}_{\rm SORT}$ is an ideal functionality, performing three sortition-related functions as shown in Fig.~\ref{fig:uc:ideal:sort}. BLOWN is denoted as $\mathcal{\pi}_{\rm B}[\mathcal{\pi}_{\rm SIG}, \mathcal{\pi}_{\rm SORT}]$, which implements real protocols $\mathcal{\pi}_{\rm SIG}$ and $\mathcal{\pi}_{\rm SORT}$. Besides, $\pi_{\rm SORT}[\mathcal{F}_{\rm VRF}]$ is a protocol that realizes sortition-related functionalities, consisting of LeaderCounter, Sortition, and VerifySortition. These functionalities are consistent with the corresponding ones specified in Algorithm~\ref{alg:utility} except that $\pi_{\rm SORT}[\mathcal{F}_{\rm VRF}]$ uses an ideal functionality $\mathcal{F}_{\rm VRF}$ in Sortition and VerifySortition. In contrast, Algorithm~\ref{alg:utility} adopts a realistic VRF implementation. Let $\mathcal{A, Z, S}$ be respectively the adversary, environment, simulator, whose specific meanings should depend on the context. We first show that the following lemma~\ref{lemma:uc} holds for $\pi_{\rm SORT}[\mathcal{F}_{\rm VRF}]$. 

\begin{lemma}
\label{lemma:uc}
    With the same security parameter $\lambda$, for each probabilistic polynomial-time (PPT) $\{\mathcal{A}, \mathcal{Z}\}$, it holds that the protocol $\pi_{\rm SORT}[\mathcal{F}_{\rm VRF}]$ securely realizes $\mathcal{F}_{\rm SORT}$ under the $\mathcal{F}_{\rm VRF}$-hybrid model. 
\end{lemma}

\begin{proof}
Let $\mathcal{A}$ be an adversary that interacts with the nodes running $\pi_{\rm SORT}[\mathcal{F}_{\rm VRF}]$ under the $\mathcal{F}_{\rm VRF}$-hybrid model. We need to construct an ideal simulator $\mathcal{S}$ such that the view of any environment $\mathcal{Z}$ of an interaction with $\mathcal{A}$ and $\pi_{\rm SORT}[\mathcal{F}_{\rm VRF}]$ is exactly the same as that of an interaction with $\mathcal{S}$ and $\mathcal{F}_{\rm SORT}$. In our construction, the simulator $\mathcal{S}$ runs $\mathcal{A}_{\mathcal{F}_{\rm VRF}}$ (under the name of $\mathcal{F}_{\rm VRF}$) and simulates other possibly involved nodes. Here, the $\mathcal{A}_{\mathcal{F}_{\rm VRF}}$ who is attacking the VRF function is identically defined as the one attacking the ideal functionality $\mathcal{F}_{\rm VRF}^{Praos}$ presented in \cite{david2018ouroboros}. $\mathcal{S}$ is responsible for forwarding messages from $\mathcal{Z}$ and $\mathcal{A}_{\mathcal{F}_{\rm VRF}}$. Besides, $\mathcal{S}$ performs the following operations:
\begin{enumerate}
    \item{Simulating value and proof generation:} When $\mathcal{S}$ receives a message ($sk_v$, seed$||$role) in the ideal process from $\mathcal{F}_{\rm SORT}$, it simulates for $\mathcal{A}_{\mathcal{F}_{\rm VRF}}$ (under the name of $\mathcal{F}_{\rm VRF}$) the process of generating (Evaluated,  $s_{id}$, $h_v$, $\pi_v$), where $s_{id}$ represents a session id which is not explicitly presented in this paper for simplicity. $\mathcal{S}$ then forwards ($h_v$, $\pi_v$) to $\mathcal{F}_{\rm SORT}$.
    \item{Simulating verification:} When $\mathcal{S}$ receives a message ($pk_v$, $h_v$, $\pi_v$, seed$||$role) in the ideal process from $\mathcal{F}_{\rm SORT}$ meaning a verificaiton query is received, it simulates for $\mathcal{A}_{\mathcal{F}_{\rm VRF}}$ the process of VRF verification. Once receiving (Verified, $s_{id}$, $h_v$, $\pi_v$, $f$), $\mathcal{S}$ forwards ($pk_v$, $h_v$, $\pi_v$, seed$||$role, $f$) to $\mathcal{F}_{\rm SORT}$. 
\end{enumerate}

It is straightforward to verify that $\mathcal{S}$ perfectly simulates the adversary and other components. That is, for any PPT $\{\mathcal{A}, \mathcal{Z}\}$, $\mathcal{Z}$ cannot distinguish between its interaction with $\mathcal{A}$ and $\pi_{\rm SORT}[\mathcal{F}_{\rm VRF}]$ or $\mathcal{S}$ and $\mathcal{F}_{\rm SORT}$. Thus one can draw a conclusion that $\pi_{\rm SORT}[\mathcal{F}_{\rm VRF}]$ securely realizes $\mathcal{F}_{\rm SORT}$ under the $\mathcal{F}_{\rm VRF}$-hybrid model. 
\end{proof}

In the setting of \cite{david2018ouroboros}, the authors elegantly proved that there exists a realistic implementation of $\pi_{\rm VRF}$ that can securely realize the ideal $\mathcal{F}_{\rm VRF}$ under the Computational Diffie-Hellman (CDH) assumption in the random oracle model. Therefore with such a secure real-world implementation, our protocol $\mathcal{\pi}_{\rm SORT}[\pi_{\rm VRF}]$, abbreviated as $\mathcal{\pi}_{\rm SORT}$, is computationally indistinguishable from $\mathcal{\pi}_{\rm SORT}[\mathcal{F}_{\rm VRF}]$, and thus securely realizes $\mathcal{F}_{\rm SORT}$ according to Lemma~\ref{lemma:uc}.
Then for the analysis of the complicated BLOWN protocol, one can get rid of the repeated reduction proofs by conducting a hybrid experiment $\mathcal{\pi}_{\rm B}[\mathcal{F}_{\rm SIG}, \mathcal{F}_{\rm SORT}]$, where $\mathcal{F}_{\rm SORT}$ is the ideal signature scheme presented in \cite{1310743}. In Section~\ref{sec:sub:persistence:liveness}, we report the salient features that can be realized by $\mathcal{\pi}_{\rm B}[\mathcal{F}_{\rm SIG}, \mathcal{F}_{\rm SORT}]$ with the ideal combinatorial functionalities $[\mathcal{F}_{\rm SIG}, \mathcal{F}_{\rm SORT}]$. Thus we need to show that the real BLOWN protocol $\mathcal{\pi}_{\rm B}[{\pi_{\rm SIG}}, {\pi_{\rm SORT}}]$ ($\pi_{\rm SIG}$ is a secure EUF-CMA digital signature scheme) and $\mathcal{\pi}_{\rm B}[\mathcal{F}_{\rm SIG}, \mathcal{F}_{\rm SORT}]$ are computationally indistinguishable so that $\mathcal{\pi}_{\rm B}[{\pi_{\rm SIG}}, {\pi_{\rm SORT}}]$ can inherit all features of $\mathcal{\pi}_{\rm B}[\mathcal{F}_{\rm SIG}, \mathcal{F}_{\rm SORT}]$. 

\begin{theorem}
\label{thm:uc}
  With the same security parameter $\lambda$, for each PPT $\{\mathcal{A}, \mathcal{Z}\}$, it holds that there is a PPT $\mathcal{S}$ such that
  \begin{equation}
    {\rm EXEC}_{\mathcal{\pi}_{\rm B}[\pi_{\rm SIG}, \pi_{\rm SORT}]}^{\mathcal{A, Z}} \approx 
    {\rm EXEC}_{\mathcal{\pi}_{\rm B}[\mathcal{F}_{\rm SIG}, \mathcal{F}_{\rm SORT}]}^{\mathcal{S, Z}},
  \end{equation}
  where ``$\approx$'' means computationally indistinguishable. 
\end{theorem}

\begin{proof}
    With a real digital signature protocol $\pi_{\rm SIG}$, we obtain ${\mathcal{\pi}_{\rm B}[\pi_{\rm SIG}, \mathcal{F}_{\rm SORT}]}$, which is a protocol under the $\mathcal{F}_{\rm SORT}$-hybrid model. From Lemma~\ref{lemma:uc}, one can see that it holds for each PPT $\mathcal{A}$ and $\mathcal{Z}$, the protocol $\pi_{\rm SORT}$ securely realizes $\mathcal{F}_{\rm SORT}$. According to the universal composition theorem, it holds that for any adversary $\mathcal{A}_{\mathcal{F}_{\rm VRF}}$, there exists an adversary $\mathcal{A}_{\mathcal{F}_{\rm SORT}}$ such that for any environment $\mathcal{Z}$, we have 
    \begin{equation}
    \label{eq:uc1}
        {\rm EXEC}_{ \mathcal{\pi}_{\rm B}[\pi_{\rm SIG}, \pi_{\rm SORT}] } ^ { \mathcal{A}_{\mathcal{F}_{\rm VRF}}, \mathcal{Z} } \approx 
        {\rm EXEC}_{ \mathcal{\pi}_{\rm B}[\pi_{\rm SIG}, \mathcal{F}_{\rm SORT}] } ^ { \mathcal{A}_{\mathcal{F}_{\rm SORT}}, \mathcal{Z}  },
    \end{equation}

    Let $\mathcal{\pi}_{\rm B}[\mathcal{F}_{\rm SIG}, \mathcal{F}_{\rm SORT}]$ be a protocol under the $\mathcal{F}_{\rm SIG}$-hybrid model with a fixed $\mathcal{F}_{\rm SORT}$. Making use of an EUF-CMA digital signature scheme $\mathcal{\pi}_{\rm SIG}$ that securely realizes $\mathcal{F}_{\rm SIG}$, we have 
    \begin{equation}
    \label{eq:uc2}
        {\rm EXEC}_{ \mathcal{\pi}_{\rm B}[\pi_{\rm SIG}, \mathcal{F}_{\rm SORT}] } ^ { \mathcal{A}_{\mathcal{F}_{\rm SIG}}, \mathcal{Z}  } \approx
        {\rm EXEC}_{ \mathcal{\pi}_{\rm B}[\mathcal{F}_{\rm SIG}, \mathcal{F}_{\rm SORT}] } ^ { \mathcal{A}_{\mathcal{F}_{\rm 0}}, \mathcal{Z}  },
    \end{equation}
 where $\mathcal{A}_{\mathcal{F}_{\rm 0}}$ is a dumb adversary. Combining (\ref{eq:uc1}) and (\ref{eq:uc2}), one can construct the simulator $\mathcal{S}$ that can run $\mathcal{A}_{\mathcal{F}_{\rm SORT}}, \mathcal{A}_{\mathcal{F}_{\rm VRF}}, \mathcal{A}_{\mathcal{F}_{\rm 0}}$ and forward messages between the adversary and $\mathcal{Z}$ so that $\mathcal{Z}$ cannot distinguish the interactions with $\mathcal{\pi}_{\rm B}[\pi_{\rm SIG}, \pi_{\rm SORT}]$ from those with $\mathcal{\pi}_{\rm B}[\mathcal{F}_{\rm SIG}, \mathcal{F}_{\rm SORT}]$. 
\end{proof}

\subsection{Persistence and liveness}
\label{sec:sub:persistence:liveness}
We first formulate a state machine $\mathbb{S}$ with the following four states: \texttt{START}, \texttt{LEADER}, \texttt{COMMIT}, \texttt{FINAL}.

\begin{definition}{(\texttt{START} State).}
The system is in \texttt{START} state when the following conditions hold: (1) $|\{v|v\in V, l_v>0\}|>1$; (2) the honest nodes that accepted $B_v^{k-1}$ in the last epoch have finished initialization. 
\end{definition}

\begin{definition}{(\texttt{LEADER} State).}
The system is in \texttt{LEADER} state when the following conditions hold: (1) there is a node $v$ with $l_v>0$ and $\forall u\in V\setminus\{v\}, l_u=0$; (2) $j=0$; (3) the size of $v$'s transaction stack $|txp_v|=0$.
\end{definition}

\begin{definition}{(\texttt{COMMIT} State).}
The system is in \texttt{COMMIT} state when the following conditions hold: (1) there is a node $v$ with $l_v>0$ and $\forall u\in V\setminus\{v\}, l_u=0$;  (2) $0<j<c\cdot i_k$.
\end{definition}

\begin{definition}{(\texttt{FINAL} State).}
The system is in \texttt{FINAL} state if one of the following two conditions holds: (1) each honest node $v$ has received a valid $m_B$ and accepted the block $B_u^{k}$; (2) honest nodes did not receive a block in the $(c+1)i_k$-th round.
\end{definition}

Garay \textit{et al.} \cite{garay2015bitcoin} proved that a secure distributed ledger should satisfy persistence and liveness properties. Let $tx^j_i$ be the $j$-th transaction of the $i$-th block (the $0$-th block is the genesis block). We say $tx^j_i$ is $t$\emph{-stable} when the current block index is larger than $i + t$, where $t>0$. Then the persistence and liveness properties that BLOWN should guarantee can be defined as follows.

\begin{definition}{\textbf{Persistence and liveness.}}
\label{thm:persistence}
   Persistence states that if an honest node $v$ proclaims a transaction $tx^j_i$ as $t$-stable, other honest nodes, if queried, either report the same result or report error messages. On the other hand, the liveness property states that if an honest node generates the transaction $tx^j_i$ and contends to broadcast it in phase two, the BLOWN protocol can add it to the blockchain within finite epochs w.h.p. 
\end{definition}

Kiayias and Panagiotakos \cite{kiayias2015speed} showed that persistence and liveness can be derived from the following three more concrete properties: chain growth, common prefix, and chain quality.

\begin{definition}[Chain growth property]
\label{def:growth}
    Consider two chains $\mathcal{C}_1, \mathcal{C}_2$ possessed by two honest nodes at the onset of two epochs $e_1<e_2$ with $e_2$ at least $k$ epochs ahead of $e_1$. It holds that $len(\mathcal{C}_2)-len(\mathcal{C}_1)\geq \tau\cdot k$, where $\tau$ is the speed coefficient with $\tau\in (0,1]$ and $ k\in \mathbb{N}$.
\end{definition}

\begin{definition}[Common prefix property]
\label{def:common:prefix}
    The chains $\mathcal{C}_1, \mathcal{C}_2$ possessed by two honest nodes at the onset of the epoch $e_1<e_2$ satisfy $\mathcal{C}_1^{\lceil k} \preceq \mathcal{C}_2$, where $k\in \mathbb{N}$ and $\mathcal{C}_1^{\lceil k}$ denotes the chain obtained by removing the last $k$ blocks from $\mathcal{C}_1$, and $\preceq$ denotes the prefix relation.
\end{definition}

\begin{definition}[Chain quality property]
\label{def:quality}
    Consider any portion of length at least $l$ of the chain possessed by an honest party at the onset of an epoch. The ratio of the blocks originated from the adversary is at most $1-\mu$, where $\mu\in (0,1]$  is the chain quality coefficient.
\end{definition}

In the remainder of this section, we prove that $\mathcal{\pi}_{\rm B}[\mathcal{F}_{\rm SIG}, \mathcal{F}_{\rm SORT}]$ satisfies chain growth, common prefix, and chain quality properties, indicating that BLOWN guarantees persistence and liveness.

\subsubsection{Chain Growth}

$\mathcal{\pi}_{\rm B}[\mathcal{F}_{\rm SIG}, \mathcal{F}_{\rm SORT}]$ meets chain growth as claimed in Theorem~\ref{thm:chain:growth}. We prove Theorem~\ref{thm:chain:growth} by two steps: 1) each epoch must be terminated within a finite time (or $\mathbb{S}$ never enters a deadlock); 2) the chain growth property should quantify the blockchain growing speed such that new blocks are added to a chain with a speed coefficient $\tau \in (0,1]$. Concretely, we first prove that $\mathcal{\pi}_{\rm B}[\mathcal{F}_{\rm SIG}, \mathcal{F}_{\rm SORT}]$ ensures robust communication channels, as without which the protocol can hardly proceed. With such a communication channel, $\mathcal{\pi}_{\rm B}[\mathcal{F}_{\rm SIG}, \mathcal{F}_{\rm SORT}]$ supports a successful leader election, which provides correctness, efficiency, and practicality. Then we perform an analysis on the $\mathbb{S}$ of BLOWN to end the proof of Theorem~\ref{thm:chain:growth}. 

\begin{theorem}
\label{thm:chain:growth}
    It holds for $\mathcal{\pi}_{\rm B}[\mathcal{F}_{\rm SIG}, \mathcal{F}_{\rm SORT}]$ that each epoch can terminate in $O(cw_{max}\lambda)$, and there are $O(cw_{max})$ transactions added to the blockchain at each epoch w.h.p., at the speed coefficient (following Definition~\ref{def:growth}) $\tau=0.5$.
\end{theorem}

To start with, we need to prove Theorem~\ref{thm:time}, which states that $\mathcal{\pi}_{\rm B}[\mathcal{F}_{\rm SIG}, \mathcal{F}_{\rm SORT}]$ can ensure a robust communication channel. Recall that the distance between any two nodes is bounded by $R_0=(P/\beta\theta)^{1/\alpha}$ in a sinlge-hop network. Therefore for $\forall v\in V$, $D_{R_0}(v)$ can cover all the neighbors of node $v$ so that if at least one node $u \in N_{R_0}(v)$ transmits a message, $v$ would either receive the message or sense a busy channel. $D_{R_0}(v)$ and $N_{R_0}(v)$ are later used for calculating aggregated transmission probability of $N_{R_0}(v)$ and the channel contention within $D_{R_0}(v)$.

\begin{theorem}\label{thm:time}
\label{thm:jamming:resistant}
  If $N_{R_0}(v)\neq \varnothing$, it holds true for BLOWN that runing at least $F=\Omega((T\log N)/\epsilon+(\log N)^4/(\gamma\epsilon)^2)$ rounds leads to at least $(1-\epsilon\beta^\prime)\rho e^{\frac{-\rho}{1-\hat{p}}}F$ rounds of successful transmissions against any $((1-\epsilon),T)$-bounded adversary w.h.p., where $\gamma=O(1/(\log T+\log \log N))$ and $\rho$ is a constant.
\end{theorem}

\begin{proof}
To prove Theorem~\ref{thm:jamming:resistant}, we divide $D_{R_0}(v)$ into six sectors of equal angles centered at $v$, and denote an arbitrary sector as $S$. Then we refer to $\bar{p_v}=\sum_{w\in S\setminus\{v\}} p_w$ as the aggregated transmission probability of the neighbors of $v$, and  $p_S$ denotes the aggregated transmission probability of all the nodes in $S$. Lemma~\ref{lemma:ps3} can be proved utilizing Lemma~\ref{lemma:ps1} and \ref{lemma:ps2}, whose proofs can be found in \cite{richa2010jamming}. We divide the $F$ into $(c\log N)/\epsilon$ consecutive subframes, with each consisting of $c(T+(\log N)^3/(\gamma^2\epsilon))$ rounds. 

\begin{lemma}
\label{lemma:ps1}
Consider any node $v$ in $S$. If $\bar{p_v}>5-\hat{p}$ during all rounds of a subframe $I'$ of $I$ and at the beginning of $I'$, $T_v\leq \sqrt{F}$, then $p_v$ is at most $1/N^2$ at the end of $I'$, w.h.p.
\end{lemma}
\begin{lemma}
\label{lemma:ps2}
For any subframe $I$ in $F$ and any initial value of $p_S$ in $I$ there is at least one round in $I$ with $p_S\leq 5$ w.h.p.
\end{lemma}
\begin{lemma}
\label{lemma:ps3}
For any subframe $I_k$ in $I$, if $p_S\leq 5$ occurs during the past subframe $I_{k-1}$, $p_S\leq 5e^2$ holds throughout $I_k$ w.m.p.
\end{lemma}
\begin{proof}
Let $p_S^t$ be the cummulative transmission probability of nodes in $S$ at round $t$.
Assume the probability that all nodes in $S$ are not transmitting is $q_0$, the probability that only one node in $S$ is transmitting is $q_1$, and the probability that at least two nodes in $S$ are transmitting is $q_2$. Then one can obtain the upper bound of the expectation of $p_S^{t+1}$ as follows:
\begin{equation}
\label{eq:ps}
\begin{aligned}
E[p_S^{t+1}]\leq q_0(1+\gamma)p_S^{t}+q_1(1+\gamma)^{-1} p_S^{t}+q_2\cdot p_S^{t}.
\end{aligned}
\end{equation}
This upper bound holds true even if we consider the rounds when $c_v>T_v$, which decreases $p_S$. Let $E_2$ be the event when at least two nodes in $S$ transmit. If $E_2$ does not happen, $q_2=0$ and Eq.~\eqref{eq:ps} becomes  
\begin{equation}
\label{eq:ps2}
\begin{aligned}
E[p_S^{t+1}]=\frac{q_0}{q_0+q_1}(1+\gamma)p_S^t+\frac{q_1}{q_0+q_1}(1+\gamma)^{-1} p_S^t.
\end{aligned}
\end{equation}
If $p_S>5$, we have $q_1 \geq p_S\cdot q_0\geq 5q_0$. Hence,
\begin{equation}
\label{eq:ps3}
\begin{aligned}
E[p_S^{t+1}]\leq [\frac{(1+\gamma)}{6}+\frac{5(1+\gamma)^{-1}}{6}]p_S^t \leq (1+\gamma)^{-1/2}p_S^t.
\end{aligned}
\end{equation}
Considering the case where $E_2$ might happen, one can rewrite $E[p_S^{t+1}]$ as 
\begin{equation}
\label{eq:ps4}
\begin{aligned}
E[p_S^{t+1}]\leq [q_2+(1-q_2)(1+\gamma)^{-1/2}]p_S^t.
\end{aligned}
\end{equation}
Since $q_2=1-q_0-q_1<1-p_Se^{\frac{p_S}{1-\hat{p}}}$, we have
\begin{equation}
\label{eq:ps5}
\begin{aligned}
E[p_S^{t+1}]\leq [1-p_Se^{\frac{p_S}{1-\hat{p}}}+(1+\gamma)^{-1/2}p_Se^{\frac{p_S}{1-\hat{p}}}]p_S^t.
\end{aligned}
\end{equation}
Suppose in the subframe $I_{k-1}$ there is a round $t$ with $p_S>5$. One can find a time interval $I'\subseteq I_{k-1}$, which satisfies $5<p_S<5e$ during $I'$, $p_S<5$ just before $I'$, and $p_S>5e$ at the end of $I'$. We intend to bound the probability at which such $I'$ happens.
Let $\phi=\log_{1+\gamma}[(1-p_Se^{\frac{p_S}{1-\hat{p}}}+(1+\gamma)^{-1/2}p_Se^{\frac{p_S}{1-\hat{p}}})^{-1}]$. Since $\gamma$ is sufficiently small, we have $\phi\in(0.5,1)$ and $E[p_S^{t+1}]\leq (1+\gamma)^{-\phi}$. On the other hand, $p_S^{t+1}\leq(1+\gamma)p_S^t\leq(1+\gamma)^{2\phi}p_S^t$. Then let $X_S^{t}=\log_{(1+\gamma)}p_S^t+\sum_{i=0}^{t-1}\phi_k$ and $X_S^0=\log_{(1+\gamma)}p_S^0$, it is easy to verify that $E[X_S^{t+1}]=X_t$ and $X_S^{t+1}\leq X_S^{t}+c_{t+1}$, where $c_{t+1}=3\phi_t$. Leveraging the Azuma–Hoeffding Inequality, it holds that 
\begin{equation}
\label{eq:ps6}
\begin{aligned}
P[X_S^{T}-X_S^0>\delta]\leq e^{\frac{-\delta^2}{2\sum_{k=1}^{T}c_k^2}},
\end{aligned}
\end{equation}
for $\delta=1/\gamma+\sum_{k=0}^{T-1}\phi_k$. Therefore
\begin{equation}
\label{eq:ps7}
\begin{aligned}
P[\log_{(1+\gamma)}p_S^T/p_S^0 > 1/\gamma]\leq e^{\frac{-\delta^2}{2\sum_{k=0}^{T-1}(3\phi_k)^2}}.
\end{aligned}
\end{equation}
Let $\psi=\sum_{k=0}^{T-1}(\phi_k)^2$, we have $e^{\frac{-\delta^2}{2\sum_{k=0}^{T-1}(3\phi_k)^2}}=\frac{(\psi+1/\gamma)^2}{18\psi}\geq\frac{1}{9\gamma}$. Hence, 
\begin{equation}
\label{eq:ps8}
\begin{aligned}
P[\log_{(1+\gamma)}p_S^T/p_S^0 > 1/\gamma]\leq e^{-1/9\gamma}\leq \frac{1}{\log^cN},
\end{aligned}
\end{equation}
for any constant $c$ if $\gamma=O(1/(\log T+\log \log N)$. Note that $\log_{(1+\gamma)}p_S^T/p_S^0 > 1/\gamma$ indicates $p_S^T/p_S^0>e$. Considering $p_S^0>5$ at the beginning of a subframe $I'$, $P[\log_{(1+\gamma)}p_S^T/p_S^0 > 1/\gamma]$ is the probability at which the aggregated probability of the nodes in $S$ exceeds $5e$ at the end of $I'$. Hence we prove that if $p_S<5$ holds at the beginning of $I_{k-1}$, $p_S<5e$ holds throughout $I_{k-1}$ w.m.p. Also, it is analogous to prove that if $p_S<5e$ is true at the beginning of $I_{k-1}$, $p_S<5e^2$ holds throughout $I_{k-1}$ w.m.p. Hence, if $p_S\leq 5$ happens during the past subframe $I_{k-1}$, $p_S<5e$ holds throughout $I_{k}$ w.m.p. Since $p_S<5e$ holds at the beginning of $I_{k}$, $p_S<5e^2$ holds throughout $I_{k}$ w.m.p., which proves the lemma.
\end{proof}

\begin{lemma}
\label{lemma:ps4}
$(1-\epsilon\beta^\prime)$-fraction of subframes in $F$ satisfy $p_V\leq \rho$ w.h.p, where $p_V=\sum_{v\in V} p_v$ is the aggregated probability of all nodes, and $\epsilon, \beta^\prime, \rho$ are constants.
\end{lemma}
\begin{proof}
Let us focus on a fixed subframe $I_k$ and its previous subframe $I_{k-1}$. Lemma~\ref{lemma:ps2} indicates that there is at least one round in $I_{k-1}$ with $p_S\leq 5$ w.h.p. Then it follows from Lemma~\ref{lemma:ps3} that if there is at least one round in $I_{k-1}$ with $p_S\leq 5$, $p_S<5e^2$ holds throughout $I_{k}$ w.m.p. Define a subframe $I$ to be \emph{good} if $p_S\leq 5e^2$ holds throughout $I$, and otherwise $I$ is \emph{bad}. Then it follows from the Chernoff bounds that at most $\epsilon\beta^\prime/6$ of the subframes in $F$ are bad w.h.p. Since $D_{R_0}$ consists of six sectors and covers all nodes in $V$, there is at least $(1-\epsilon\beta^\prime)$-fraction of subframes in which the aggregated probability $p_V=\sum_{v\in V} p_v=\sum_{v\in N_{R_0}(v)} p_v$ is bounded by $\rho=6\times5e^2=30e^2$, which completes the proof.
\end{proof}
Then, the probability on which there exists one successful transmission is given by
\begin{equation}
\begin{aligned}
\sum_{v\in V} p_v \prod_{w\in V\backslash v}(1-p_w) &\geq \sum_{v\in V} {p_v\prod_{w\in V}}(1-p_w)\\
&\geq \sum_{v\in V} {p_v\prod_{w\in V}} e^{\frac{-p_w}{1-\hat{p}}}\\
&= \sum_{v\in V} p_v e^{\frac{-p_V}{1-\hat{p}}}\\
&= p_Ve^{\frac{-p_V}{1-\hat{p}}}\\
&\geq \rho e^{\frac{-\rho}{1-\hat{p}}}.\\
\end{aligned}
\end{equation}
\end{proof}

With the robust communication guarantee, we next prove that BLOWN can support a successful leader election, which is the core of the protocol. 
Most leader election algorithms in wireless networks are only responsible for reaching the state at which one node is the leader and others are followers. Our algorithm goes one-step further by ensuring that all nodes have an identical view of the network after leader election, which is crucial to the our protocol, as shown in Theorem~\ref{thm:leader:success}.

\begin{theorem}
\label{thm:leader:success}{(Successful leader election).} Let $w_{max}$ be the maximum weight among all nodes and $\lambda$ be a constant to be determined.
   $\mathcal{\pi}_{\rm B}[\mathcal{F}_{\rm SIG}, \mathcal{F}_{\rm SORT}]$ ensures a successful leader election while satisfying the following three properties: 1) \textbf{Correctness}: only one node is left as the leader with a positive $l_v$ at the end of $P_1$; 2) \textbf{Efficiency}: the success of leader election can be achieved with $O(w_{max})$ successful transmission; 3) \textbf{Practicality}: the leader and the followers should have the knowledge regarding who is the leader and at which round the leader is elected.
\end{theorem}

\begin{proof}
We prove the three properties in order. During a leader election process, all nodes contend for broadcasting messages in $P_1$ until only one node is left with a positive $l_v$, which can always be achieved inevitably. This can be proved by contradiction. Without loss of generality, we assume that there are two nodes left with a positive $l_v$. If these two nodes broadcast messages at the same round, they can not receive messages from each other simultaneously. Therefore, there is no chance for two nodes to receive messages in the same round, and there must be only one node surviving at the end. One can trivially expand this result to the cases with $3,4,\cdots,N$ nodes left with positive $l_v$ values, thus proving that the protocol can always lead to the state when only one node survives as the leader with a positive $l_v$.


To prove the efficiency property, we resort to Theorem~\ref{thm:jamming:resistant}, which shows that a constant fraction of the rounds have successful transmissions w.h.p. Concretely, a successful communication should happen once every $\lambda=(1-\epsilon\beta)^{-1}\rho^{-1}e^{\frac{\rho}{1-\hat{p}}}$ rounds on average w.h.p. Then leader election can be finished in $O(w_{max}\lambda )$ rounds w.h.p. This indicates that $O(w_max)$ number of successful transmissions can lead to a successful leader election and the communication complexity is not directly related to the network size.

To prove the practicality, we denote $E_v$ as the event that $v$ broadcasts a message in slot one and senses an idle channel in slot two. In this case, $v$ would know itself as the leader. 
Let $p_v$ be the probability that $v$ broadcasts a message in slot one, $p_v^{(0)}$ be the probability that $v$ broadcasts a message and there is also at least one node $u$ with $l_u>0$ broadcasting a message in slot one, $p_v^{(1)}$ be the probability that $v$ broadcasts a message and there exists at least one node $u$ with $l_u>0$ sensing the channel in slot one, and $p_v^{(2)}$ be the probability that $v$ broadcasts a message in slot one and $l_u=0, \forall u\in V\backslash \{v\}$. Certainly, $p_v=p_v^{(0)}+p_v^{(1)}+p_v^{(2)}$. 
If $E_v$ happens, $v$ senses an idle channel in slot two. Then $p_v^{(0)}=0$ since if $u$ broadcasts a message in slot one, a follower $f$ senses interference and thus broadcasts an $m$ in slot two so that $v$ senses interference in slot two, which contradicts our assumption. Also, $p_v^{(1)}=0$ because if there exists a node $u$ with $l_u>0$ sensing the channel in slot one, $u$ has to broadcast a message in slot two which also contradicts the assumption. Therefore, we obtain the result that if $E_v$ happens, $v$ can confirm itself as the unique leader. 

Correspondingly, we denote $E_f$ as the event that a follower $f$ recognizes $v$ as the leader when $f$ receives a message from $v$ and obtains $\mathcal{I+N}<\theta$ in slot one, then senses an idle channel in slot two. 
Let $p_f$ be the probability that $f$ receives a message from $v$ and obtains $\mathcal{I+N}<\theta$ in slot one, $p_f^{(0)}$ be the probability that there is at least one node $u\in V\backslash \{v\}$ with $l_u>0$ sensing the channel in slot one, $p_v^{(1)}$ be the probability that $v$ is the unique leader; then we have $p_f=p_f^{(0)}+p_f^{(1)}$. Assume $E_v$ happens, we have $p_f^{(0)}=0$ since if $p_f^{(0)}\neq 0$, $u$ has to broadcast a message in slot two and thus a follower senses interference, which contradicts our assumption. As a result, $E_f$ indicates that $v$ is the unique leader. Additionally, the round at which a successful leader election happens can be found when $E_v$ and $E_f$ occur simultaneously, which ends the proof of the third property. 
\end{proof}

Utilizing Theorem~\ref{thm:time} and \ref{thm:leader:success} as intermediate conclusions, we are finally ready to prove  Theorem~\ref{thm:chain:growth}.

\begin{proof}
   The time between \texttt{LEADER} and \texttt{COMMIT} is fixed to $j=c\cdot i_k$ rounds, where $c$ is an adjustable constant parameter according to different implementation scenarios. If the leader does not broadcast a block in the $(c+1)i_k$-th round, the state transits to the final state since $\mathbb{S}$ satisfies the second condition of a \texttt{FINAL} state. 
   Then $\mathbb{S}$ starts the next epoch. According to Theorem~\ref{thm:leader:success}, each epoch can be terminated in $O(cw_{max}\lambda)$ w.h.p., and there should be $O(cw_{max})$ transactions added to the blockchain in each epoch w.h.p. 
   
   Assume an honest node $v$ generates the transaction $tx^j_i$ and contends to broadcast it in $P_2$. The transaction can be received by an honest leader with probability at least $p=cw_{max}/N$ in each epoch. By applying the Chernoff bound, we obtain that $tx^j_i$ can be added to the blockchain within $n$ epochs with probability at least $1-e^{-\frac{(np-1)^2}{2}}$, where $n$ is the number of epochs when $v$ broadcasts $tx^j_i$. The above analysis indicates that $\mathbb{S}$ has no chance of staying at a  deadlock in any epoch. Considering the assumption that honest nodes control more than $50\%$ coins, $\mathcal{\pi}_{\rm B}[\mathcal{F}_{\rm SIG}, \mathcal{F}_{\rm SORT}]$ with ideal functionalities $[\mathcal{F}_{\rm SIG}, \mathcal{F}_{\rm SORT}]$ can ensure a fair sortition based on the nodes' coin distribution. Thus, with probability at least $50\%$, an honest node can be selected as a leader to propose a new block. When  two chains $\mathcal{C}_1, \mathcal{C}_2$ possessed by two honest nodes at the onset of two epochs $e_1<e_2$ with $e_2$ at least $k$ epochs ahead of $e_1$, it holds that $len(\mathcal{C}_2)-len(\mathcal{C}_1)\geq \tau\cdot k$, where $\tau=0.5$. This completes the proof of Theorem~\ref{thm:chain:growth}.

\end{proof}

\subsubsection{Common Prefix}
\label{sec:subsub:common:prefix}

\begin{theorem}[]
\label{thm:common:prefix}
    $\mathcal{\pi}_{\rm B}[\mathcal{F}_{\rm SIG}, \mathcal{F}_{\rm SORT}]$ satisfies the common prefix property (following Definition~\ref{def:common:prefix}).
\end{theorem}

\begin{proof}
In $\mathcal{\pi}_{\rm B}[\mathcal{F}_{\rm SIG}, \mathcal{F}_{\rm SORT}]$ a node can directly append a new block $B_u^{k}$ to its local blockchain only when $\mathcal{F}_{\rm SIG}$ and $\mathcal{F}_{\rm SORT}$ answer with 1 when being queried. Therefore, an adversary who intends to disguise itself as a leader to propose a block should fail since it cannot break $[\mathcal{F}_{\rm SIG}$ and $\mathcal{F}_{\rm SORT}]$. However, a malicious leader (an adversary who wins the leader election) can still diverge the global distributed ledger to cause a $\Delta$-fork defined in Definition~\ref{def:fork}. 

\begin{figure}[!htbp]
\centering
\includegraphics[width=3.2in]{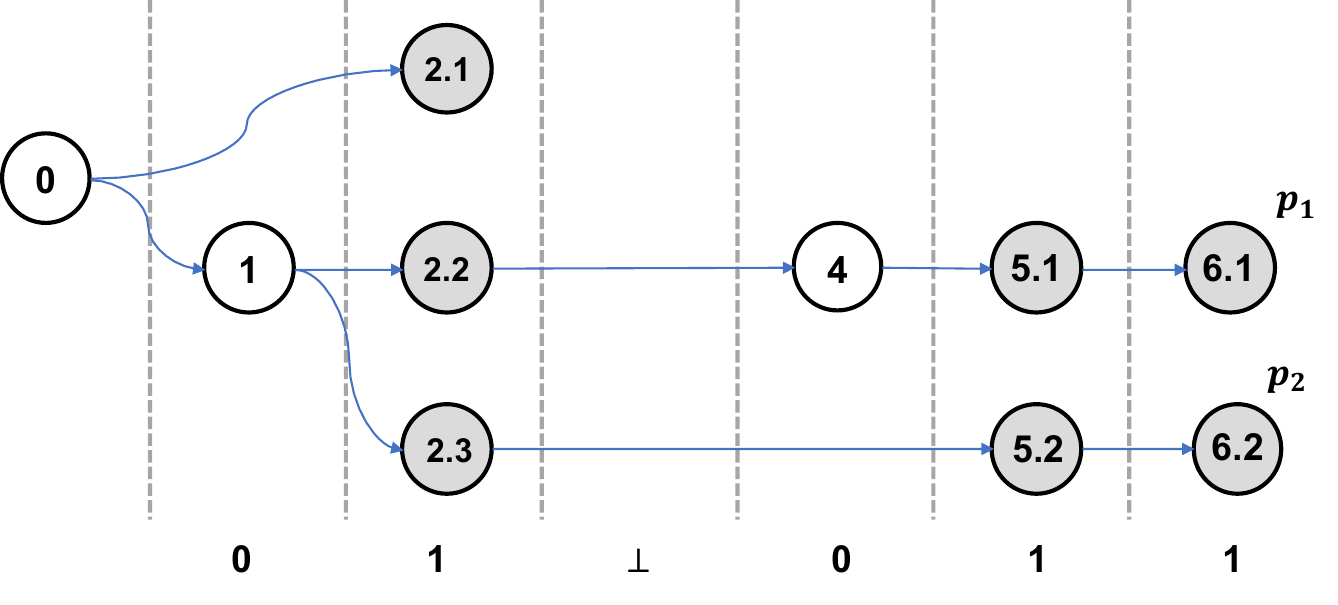}
\caption{
A simple example to illustrate the concepts of string, $\Delta$-fork, and divergence. Specifically, we have $s_6=\{0,1,\perp,0,1,1\}$, $l(p_1)=5$, and $l(p_2)=4$. The tree presented here is a $\Delta$-fork with $\Delta=1$, and $|div(p_1, p_2)|=4$. 
}
\label{fig:fork}
\end{figure}

\begin{definition}[String]
    Consider an epoch $e_k$ during the execution of functionality $\mathcal{\pi}_{\rm B}[\mathcal{F}_{\rm SIG}, \mathcal{F}_{\rm SORT}]$ with adversary $\mathcal{A}$ and environment $\mathcal{Z}$. Let $S=\{e_1, \cdots, e_k\}$ denote a sequence of epochs of length $k$. The string $s=\{0,1, \perp\}^k$ of $S$ is defined so that $s_i=1$ if the adversary controls the epoch leader of $e_i$ and broadcasts a block, $s_i=0$ if an honest node controls the epoch and broadcasts a block, and $s_i=\perp$ if no block is broadcast. We say that the index $i$ is adversarial if $s_i=1$ and honest otherwise. 
\end{definition}

W.l.o.g., let $s_0=0$ for $e_0$ meaning that the genesis block has an honest index.

\begin{definition}[$\Delta$-Fork]
\label{def:fork}
    Let string $s=\{0,1, \perp\}^k$ of $S$ and $\Delta$ be a non-negative interger. A $\Delta$-fork is a directed, acyclic, rooted tree $F=(V, E)$  in which  the two longest paths $p_1$ and $p_2$ satisfy $|l(p_1)-l(p_2)|\leq \Delta$, where a path $p$ refers to a road from the root to a leaf and $l(p)$ is the hop-count (length) of the path $p$. 
\end{definition}

\begin{definition}[Divergence]
    Denote the divergence of two paths $p_1$ and $p_2$ in a $\Delta$-Fork as $div(p_1, p_2)$, which is defined as 
    \begin{equation}
        div(p_1, p_2)=\max \{ l(p_1), l(p_2)\}-l(p_1\cap p_2),
    \end{equation}
    where $l(p_1\cap p_2)$ is the legnth of the common path of $p_1$ and $p_2$, and $div(p_1, p_2)$ is non-negative.
\end{definition}

\begin{lemma}
    The common prefix property  is satisfied if and only if for any pair of paths $p_i, p_j, i\neq j$, in a $\Delta$-fork, $div(p_i, p_j)\leq k$. 
\end{lemma}
\begin{proof}
    For the ``only if'' direction, we assume that there exits a path $p_1, p_2$ (w.l.o.g., $l(p_1)>l(p_2)$) such that $div(p_1, p_2)> k$. That is $\max \{ l(p_1), l(p_2)\}-l(p_1\cap p_2)=l(p_1)-l(p_1\cap p_2)>k$. Let $V_1$ ($V_2$) be the set of honest nodes that store the distributed ledger as the path $p_1$ ($p_2$). Once querying a local blockchain, any $v_1\in V_1$ ($v_2\in V_2$) responds with $\mathcal{C}_1$ ($\mathcal{C}_2$). Denote the latest point of the common path $p_1\cap p_2$ as $\hat{v}$, which is also called a bifurcation point. The path $p_1^{\lceil k}$ that is obtained by truncating the last $k$ vertices of $p_1$ still covers $\hat{v}$, which is not the endpoint of $p_1$ since $l(p_1)-k>l(p_1\cap p_2)$. Denote the endpoint of $p_1$ as $end(p_1)$. Then the blocks corresponding to the points from $\hat{v}$ to $end(p_1)$ are included in $\mathcal{C}_1^{\lceil k}$, but the block mapped to $end(p_1^{\lceil k})$ is not included in $\mathcal{C}_2$, thus violating the common prefix property. For the ``if'' direction, assuming that the common prefix is violated, there exists a pair of ledgers $\mathcal{C}_1$ and $\mathcal{C}_2$ for $e_1<e_2$ such that $\mathcal{C}_1^{\lceil k} \npreceq \mathcal{C}_2$. Mapping such blockchains to two distinct paths $p_1, p_2$, the endpoint $end(p_1^{\lceil k})$ corresponding to the latest block in $\mathcal{C}_1^{\lceil k}$ is not covered by $p_2$ and comes after $\hat{v}$. By the definition of divergence, $div(p_1, p_2)=\max \{ l(p_1), l(p_2)\}-l(p_1\cap p_2)>k$.
\end{proof}

Here one can define a common prefix violation as the case when there exit two paths $p_1, p_2$ in a $\Delta$-fork with $|div(p_1, p_2)|>k$. To prove Theorem~\ref{thm:common:prefix}, we need to show that a common prefix violation happens with an extremely small probability. Generally speaking, $\Delta\leq k$, and $p_1, p_2$ can be regarded as the respective paths that the honest nodes and adversary go through. This is based on the assumption that all honest nodes strictly follow the longest chain rule, while the adversary focuses on increasing the length of an illegal chain (e.g., including a double-spend transaction). Therefore, a common prefix violation can also be interpreted as a race between honest nodes and the adversary that lasts for more than $k$ blocks, but their view paths still follow $|l(p_1)-l(p_2)|\leq \Delta$. Let $X_i\in \{\pm 1\}$ (for $i=1, 2, \cdots$) denote a series of independent random variables for which $Pr[X_i=1]=(1-\epsilon)/2$. Note that $\epsilon\in(0,1)$ is satisfied in functionality $\mathcal{\pi}_{\rm B}[\mathcal{F}_{\rm SIG}, \mathcal{F}_{\rm SORsT}]$ since the adversary controls less than 50\% coins and the protocol adopts a hybrid $[\mathcal{F}_{\rm SIG}, \mathcal{F}_{\rm SORT}]$ to ensure that the probability of the adversary being a leader is less than 1/2. Consider $k$ epochs of the biased walk beginning at the bifurcation point. The resulting value is tightly concentrated at $-\epsilon k$. By applying the Chernoff bound, for each $k$ random walk hitting problem, we have 
\begin{equation}
    Pr[X_k<\Delta]\leq e^{-(1-\Delta/\epsilon k)^2\epsilon k/2}=e^{-O(k)},
\end{equation}
where $\Delta \ll k$. This indicates that $\mathcal{\pi}_{\rm B}[\mathcal{F}_{\rm SIG}, \mathcal{F}_{\rm SORT}]$ satisfies the common prefix property w.h.p., which completes the proof.
\end{proof}

\subsubsection{Chain Quality}
\label{sec:sub:sybil}

The chain quality property requires that a certain fraction of the blocks should satisfy high quality standards (high-quality blocks are the ones generated absolutely by honest nodes). Chain quality can be threatened by Sybil attacks which are particularly harmful  in wireless networks \cite{newsome2004sybil}. In a Sybil attack, an attacker can behave as many nodes by illegitimately claiming massive identities or impersonating others. A successful attacker chosen as a leader can deny to broadcast a new block or broadcast an invalid block. Since honest nodes can neither wait for more than $c\cdot i_k$ rounds in $P_2$ nor accept invalid blocks, the attacker cannot hinder the system from changing from \texttt{LEADER} state to the \texttt{FINAL} state. However, an attacker can make an epoch wasted without any new block being added to the blockchain, thereby harming the liveness. Our BLOWN protocol prevents Sybil attacks and ensures liveness under the assumption that all malicious nodes control no more than $50\%$ coins of the entire network. 

Consider one epoch. $\mathcal{F}_{\rm SORT}$ provides a binomial distribution as $B(k;w_v,p)=\binom{w_v}{k}p^k(1-p)^{w_v-k}$, which has a salient property that splitting coins into multiple sub-users does not give attackers any advantage. In particular, suppose an attacker splits its account balance $w_\mathcal{A}$ into $w^1_\mathcal{A}, w^2_\mathcal{A}, \cdots, w^n_\mathcal{A}$, thus each sub-user has a binomial distribution as $X^i_\mathcal{A}\sim B(w^i_\mathcal{A},p)$. However, splitting coins does not increase the sum of the values of the leader counter controlled by the attacker since $(X^1_\mathcal{A}+X^2_\mathcal{A}+\cdots+X^n_\mathcal{A})\sim B(w^1_\mathcal{A}+w^2_\mathcal{A}+\cdots+w^n_\mathcal{A},p)$. Also, splitting coins decreases the maximum of the leader counter of the sub-users, which makes it harder for a sub-user to survive in $P_1$. Without loss of generality, suppose each node has an equal value of balance. Then at each epoch, the probability of a malicious node being chosen as a leader is no more than $50\%$.

\begin{theorem}\label{thm:sybil}
  Given that the ratio of the adversarial coins $\alpha<1/2$, $\mathcal{\pi}_{\rm B}[\mathcal{F}_{\rm SIG}, \mathcal{F}_{\rm SORT}]$ satisfies the chain quality property with $\mu=1-(1+\delta)\alpha$, where $\delta\in(0,1)$.
\end{theorem}

\begin{proof}
Let $X_i$ denote the event where the $i$th epoch has an adversarial leader.  We have $\mathbb{E}[X_i]\le \alpha l$. Applying the Chernoff bound we obtain
\begin{equation}
\begin{aligned}
    P_{r}[X\geq (1+\delta)\alpha l]\leq e^{-O(l)}.
\end{aligned}
\end{equation}
Then the probability that the ratio $\beta$ for the blocks originated from the adversary is at most $(1+\delta)\alpha$ is given as 
\begin{equation}
\begin{aligned}
    P_{r}[\beta\leq(1+\delta)\alpha]=1-P_{r}[X\geq (1+\delta)\alpha l]\geq1-e^{-O(l)}.
\end{aligned}
\end{equation}
When $l$ is sufficiently large, $\beta\leq(1+\delta)\alpha$ w.h.p. Thus we complete the proof of the chain quality property with $\mu=1-(1+\delta)\alpha$. Note that even though $\mu=1-(1+\delta)\alpha$ blocks can be proposed by the adversary, these blocks only contain a small fraction of malicious ones (jointly ensured by the chain growth and common prefix properties). 
\end{proof}

Therefore we can conclude that $\mathcal{\pi}_{\rm B}[\mathcal{F}_{\rm SIG}, \mathcal{F}_{\rm SORT}]$ satisfies the chain growth, common prefix, and chain quality properties, thus guaranteeing persistence and liveness. By applying Theorem~\ref{thm:uc}, BLOWN (i.e., $\mathcal{\pi}_{\rm B}[\mathcal{\pi}_{\rm SIG}, \mathcal{\pi}_{\rm SORT}]$) naturally ensures persistence and liveness.

\section{Simulation Study}
\label{sec:simulations}


In this section, we implement a simulator to investigate how various parameters impact the performance of our BLOWN protocol. Specifically, in Section~\ref{subsec:sim:throughput}, we first demonstrate the correctness and efficiency of BLOWN by considering its convergence behavior as well as its performance when network size and density vary. Then we present the performance of BLOWN  under various jamming and Sybil attack scenarios in Section~\ref{subsec:jamming:attacks}. Note that the convergence behavior of BLOWN needs to be examined from a microscopic perspective and thus we consider a single epoch; while the performance of BLOWN should be explored from a macroscopic perspective and thus multiple epochs are considered.

In our simulation, we use the crypto library of golang\footnote{https://github.com/golang/crypto} and adopt ed25519 for digital signatures, with 64-byte private key, 32-byte public key, and 64-byte signature. Public keys are broadcast to all nodes on the preset of our simulations. Besides, The Sortition algorithm is implemented with the VRF provided by CONIKS\footnote{https://github.com/coniks-sys/coniks-go/tree/master/crypto/vrf}. We employ two types of 2-dimensional planes of size $d\times d$ units, where $d=10$ or $d=\sqrt{N}$ with $N$ known as the network size. Nodes are randomly generated and distributed in the plane and no two nodes can have the same coordinates. The unit for an epoch length is round. If not stated otherwise, we adopt the following parameters $\alpha=4$, $\beta=2$, $\theta=2$, $P=\beta\theta (\sqrt{2}d)^\alpha$, $\hat{p}=0.1$, $T=60$, $\epsilon=0.3$, $\gamma=0.1$, $w_v=20$, $\tau = W/2$, and $c=10$. Besides, nodes are uniformly distributed and the percentage of the Sybil nodes is 0\% by default. Without loss of generality, all parameter values are chosen carefully to reflect various real-world cases, but not to aim to optimize the performance. All the experiments are performed under a CentOS 7 operating system running on a machine with an Intel Xeon 3.4 GHz CPU, 120 GB RAM and 1 TB SATA Hard Drive. All the reported results are the average of 100 runs, unless stated otherwise. 

The performance metrics under our consideration include throughput and average epoch length. We choose epoch length as a performance metric since it depicts how many rounds BLOWN takes to accept or discard a block. Denote by $|txp_t|$ the number of transactions received by the leader within $t=i+j$ rounds, with $i$ and $j$ respectively being the number of rounds in $P_1$ and $P_2$. Given that the unit slot time for IEEE 802.11 is set to be $50\mu s$, we have throughput as
\begin{equation}
\text{Throughput} = \frac{|txp_t|}{i\times 100\mu s + j\times 50\mu s}
\end{equation}
since $r^{k}_{1,i}$ has two slots while $r^{k}_{2,j}$ has only one slot.

\subsection{Correctness and Efficiency}
\label{subsec:sim:throughput}

\begin{figure}[!htbp]
\centering
\includegraphics[width=3.2in]{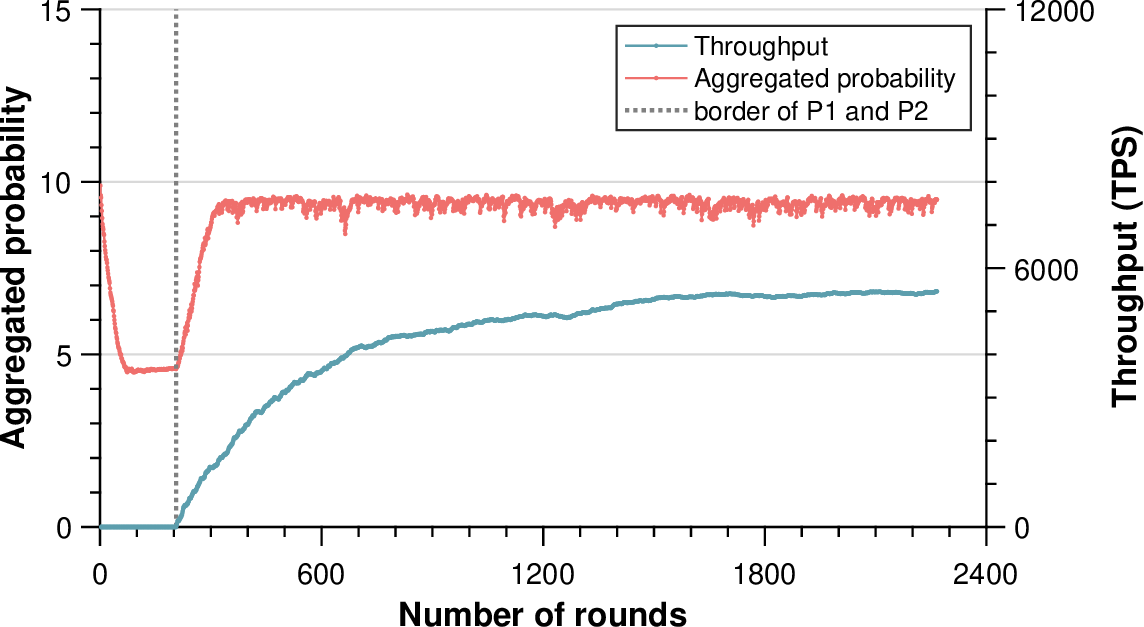}
\caption{Aggregated probability and throughput \emph{vs.} number of rounds, to demonstrate the convergence behavior, where $density=1, d=10, N=100$, and nodes are uniformly distributed.}
\label{fig:com}
\end{figure}

In this subsection we first demonstrate the convergence behavior of BLOWN then report its performance when network size and density vary. 

\textbf{Convergence Study.} 
Fig.~\ref{fig:com} presents a typical example to illustrate the convergence of the aggregated probability $p_V=\sum_{v\in V} p_v$ and throughput during one-epoch execution, where $p_V=N\times\hat{p}=10$ in the outset. There is a gray dash borderline distinguishing $P_1$ and $P_2$. Since BLOWN can rapidly adjust the initial parameters by multiplicatively increasing or deceasing $p_v$, $p_V$ adapts rapidly to reduce the noise in the channel to help achieve successful communications. Therefore, it only takes 206 rounds (corresponding to 0.206s in a real-world setting) to complete $P_1$. Such a quick adaptation contributes to the throughput of the entire protocol. In $P_2$, nodes all become active to broadcast transactions enabling $p_V$ to grow. The leader collects transactions from the 207th to the 2265th round, and a block is finalized at the 2266th round. Note that $p_V$ and throughput respectively converge to 5399 TPS and 9.37, which are mean values calculated from the last 500 rounds. Besides, we evaluate cryptographic overhead (in ms, an average of 1000 repeated trials), including the overhead of signing a transaction (0.09 ms), verifying a transaction [0.21 ms], signing a block (1.20 ms), confirming a block (930.14 ms), Sortition(3.02 ms), and VerifySortition(4.57 ms).

\begin{figure*}[!t]
\centering
\subfigure[Epoch length and throughput \emph{vs.} the network size $N$, where $density=1, d=\sqrt{N}\times\sqrt{N}$.]{
    \label{fig:size}
    \centering
    \includegraphics[width=3.2in]{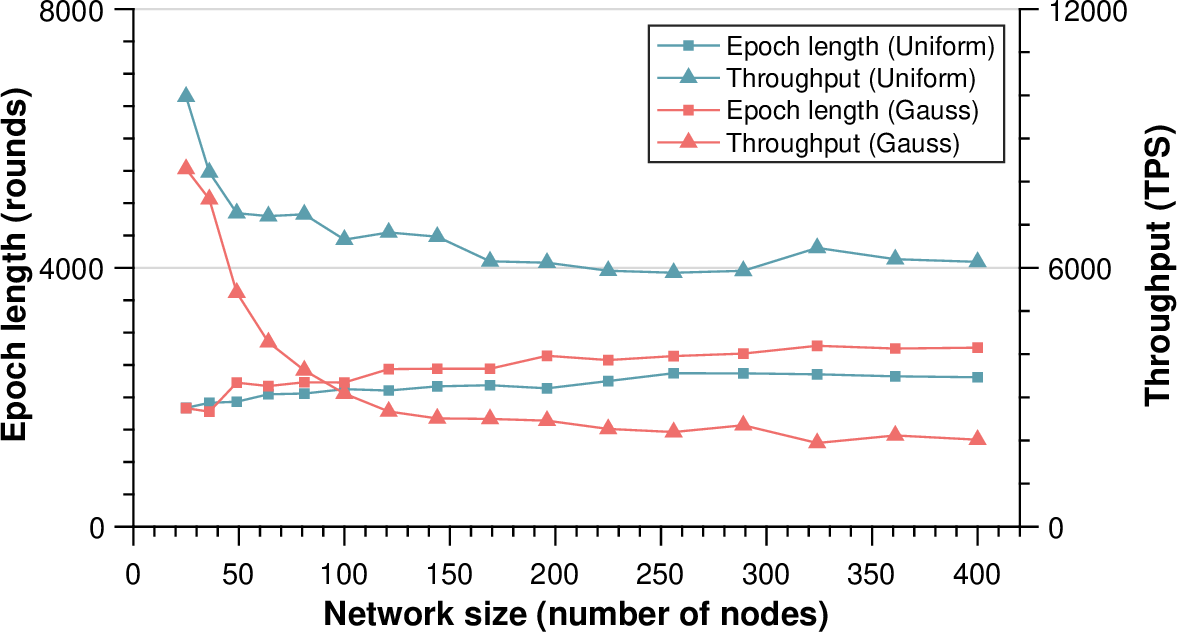}
}%
\quad
\subfigure[Epoch length and throughput \emph{vs.} the density, where $d=10, N=100$.]{
    \label{fig:density}
    \centering
    \includegraphics[width=3.2in]{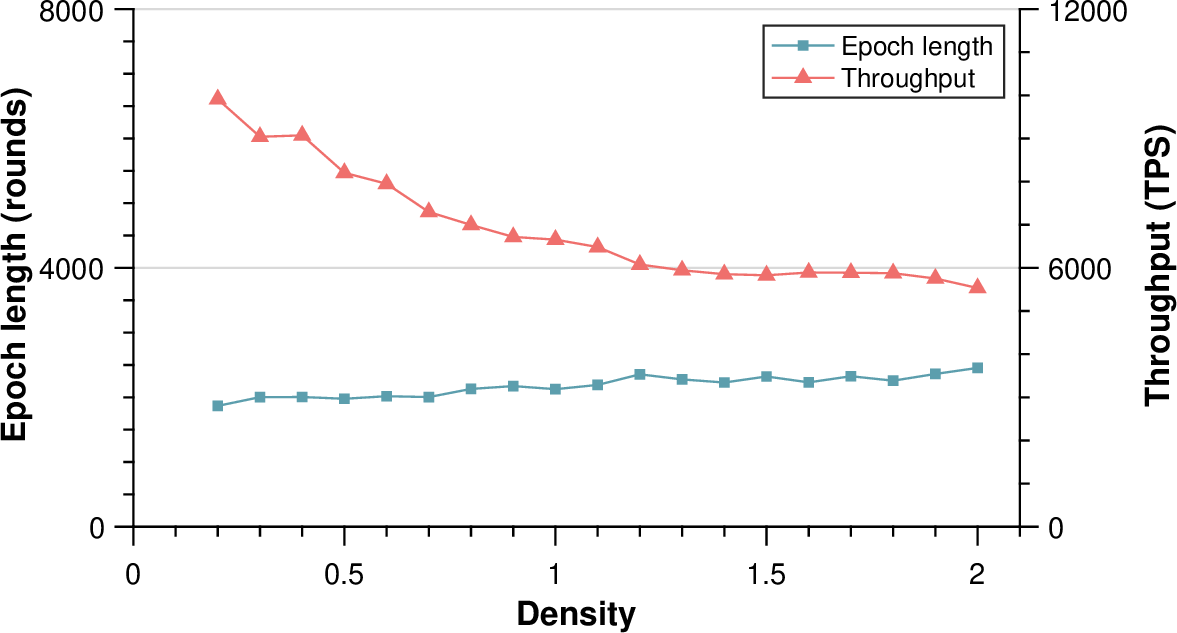}
}%
\caption{The performance of BLOWN \emph{vs.} network size (in a uniform or Gauss distribution) and density.}
\end{figure*}

\textbf{Performance \emph{vs.} Network Size.} Next we simulate the performance as a function of the network size (or $N$), where nodes are scattered in the plane of size $d=\sqrt{N}\times\sqrt{N}$ following a uniform or Gauss distribution. As shown in Fig.~\ref{fig:size}, the epoch length slowly increases with a larger $N$ with both uniform and Gauss distributions, which also means that the leader election costs more time for a larger $N$. On the other hand, throughput decreases with a larger $N$ since the added nodes lead to heavier contention. However, because of the resiliency of our jamming resistant channel, throughput can converge to about 6000 TPS and 2000 TPS for the uniform and Gauss distribution, respectively.
Compared with the uniform distribution, Gauss distribution always has a larger epoch length and lower throughput since denser nodes centrally aggregate, leading to stronger contention.

\textbf{Performance \emph{vs.} Network Density.} We also investigate how the network density impacts on the performance of the BLOWN protocol. Nodes are uniformly distributed in a $10\times 10$ plane, and $density = 0.2, 0,3, \cdots, 2$. As shown in Fig.~\ref{fig:density}, the epoch length slowly increases from the 1867 to the 2464 rounds, with the density rising tenfold. The throughput decreases for larger density and approximately converges to 6000 TPS.

\subsection{Jamming Attacks and Sybil Attacks}
\label{subsec:jamming:attacks}

\begin{figure*}[!htbp]
\centering
\subfigure[Epoch length and throughput \emph{vs.} $\epsilon$, where $density=1, d=10, N=100$.]{
    \label{fig:jammer}
    \centering
    \includegraphics[width=3.27in]{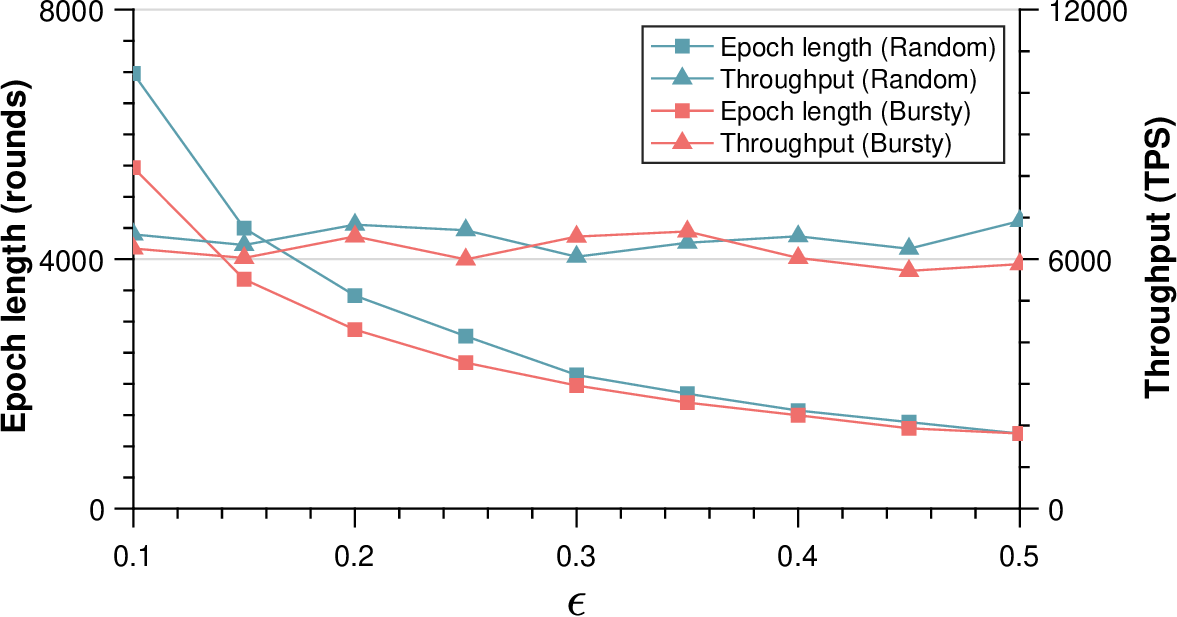}
}%
\quad
\subfigure[Epoch length and throughput with the percentage of Sybil nodes, where $density=1$, $d=10$, and $N=100$.]{
    \label{fig:sybil}
    \centering
    \includegraphics[width=3.2in]{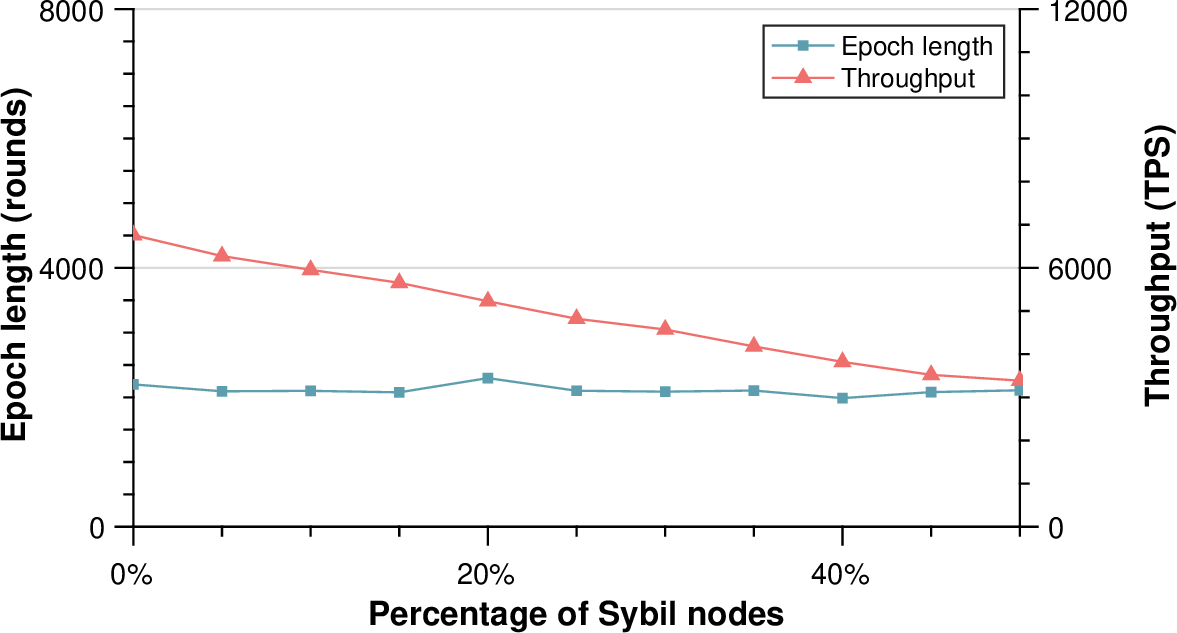}
}%
\caption{The performance of BLOWN when confronting jamming attacks (random jammers or bursty jammers) and Sybil attacks.}
\end{figure*}

\textbf{Jamming Attacks.} Here we present our protocol's performance when confronting jammers who can choose different strategies with the constraint of $(1-\epsilon)T$. We consider two types of jammers: \emph{random} jammers that can randomly jam $(1-\epsilon)T$ rounds at any interval of length $T$ and \emph{bursty} jammers who would jam  $(1-\epsilon)T$ consecutive rounds at any interval of length $T$. We test the epoch length and throughput when $\epsilon={0.1,0.15,\cdots,0.5}$, with a higher $\epsilon$ implying a lower attack frequency. The results are demonstrated in Fig.~\ref{fig:jammer}, which indicate that the epoch length decreases with the increasing $\epsilon$ due to the lower frequency of jamming attacks. Besides, $\epsilon$ does not significantly impact the throughput for both kinds of jammers. The epoch length increases faster with lower $\epsilon$ considering random jammers, indicating that random jammers are more powerful than bursty ones in BLOWN. This is because the introduction of $T_v$ makes it easier to address continuous heavy contentions.

\textbf{Sybil Attacks.} In a Sybil attack, an attacker can control massive malicious nodes that compete for being a leader but refuse to collect transactions and propose blocks. In this circumstance, the epoch with a malicious leader would be abandoned so that there is no valid block to be accepted within such an epoch. Even though we already show in our protocol analysis that BLOWN can defend against Sybil attackers who control less than 50\% wealth of the entire network, such attackers can harm the liveness of our protocol. In Fig.~\ref{fig:sybil}, the percentage of Sybil nodes does not impact the epoch length since Sybil nodes are not absent from competing in the leader election. However, the throughput has an evident linear decline for a larger percentage of Sybil nodes. Compared to the setting without Sybil nodes, $50\%N$ Sybil nodes would decrease the throughput by 49.90\%.

\section{Conclusion and Future Research}
\label{sec:conclusion}

In this paper, we propose a 2-phase blockchain protocol, namely BLOWN. BLOWN establishes a jamming-resistant communication channel and combines the Sortition algorithm and our newly proposed PoC consensus algorithm for efficient and secure leader election. Besides, BLOWN prevents double-spending attacks and Sybil attacks. Analysis and simulation results demonstrate the efficiency, effectiveness, and security properties of the BLOWN protocol. In our future research, we will investigate the multi-hop version of BLOWN, as well as the Byzantine fault-tolerant BLOWN in wireless ad hoc or fading channel settings. Also, it is neccessary to explore how practical attacks such as eclipse attacks, nothing-at-stake attacks, selfish-mining attacks can be mitigated by our protocol.

\section*{Acknowledgment}

This study was partially supported by the National Natural Science Foundation of China under grants 61832012, 61871466, 61771289 and 61672321, the Blockchain Core Technology Strategic Research Program of Ministry of Education of China under grant 2020KJ010301, and the Key Science and Technology Project of Guangxi under grant AB19110044. 

\bibliographystyle{IEEEtran}
\bibliography{references}

\end{document}